\title{Classification of integrable discrete equations of octahedron type}
\author{{Vsevolod E. Adler$^{1,2}$, Alexander I. Bobenko$^2$, Yuri B. Suris$^3$}}
\date{15 November 2010}
\documentclass[11pt]{article}
\usepackage{amsmath,amssymb,amsthm,epic,eepicemu}
\usepackage[pdftex]{graphicx}

\usepackage[width=0.9\textwidth,indention=5mm,justification=centerlast,
 font=small,labelfont=bf,labelsep=period]{caption}

\usepackage[pdfpagemode=UseOutlines,colorlinks,linkcolor=blue,bookmarksopen=true,
urlcolor=blue,unicode=true,bookmarksnumbered=true]{hyperref}

\advance\textwidth 24mm  \advance\oddsidemargin -12mm
\advance\textheight 25mm \advance\topmargin -20mm

\def\a{\alpha}
\def\b{\beta}
\def\g{\gamma}
\def\d{\delta}
\def\eps{\varepsilon}
\def\la{\lambda}
\def\ee{\mathfrak e}

\def\Integer{\mathbb{Z}}
\def\const{\mathop{\rm const}}
\def\id{\mathop{\rm id}}

\def\Ifrac#1#2{\genfrac{(}{)}{}{}{1}{#1}_{\!\!#2}}
\def\EQ#1{\langle#1\rangle}

\theoremstyle{plain}

\newtheorem{theorem}{Theorem}
\newtheorem{statement}[theorem]{Proposition}
\newtheorem{lemma}[theorem]{Lemma}

\newtheorem{definition}{Definition}

\theoremstyle{remark}

\newtheorem{remark}{Remark}

\allowdisplaybreaks[4]

\begin{document}\maketitle\thispagestyle{empty}

\footnotetext[1]{L.D. Landau Institute for Theoretical Physics, 1a Semenov pr.,
 142432 Chernogolovka, Russia. E-mail: adler@itp.ac.ru}
\footnotetext[2]{Institut f\"ur Mathematik, MA 8-3, Technische
 Universit\"at Berlin, Str. des 17. Juni 136, 10623 Berlin, Germany.
 E-mail: bobenko@math.tu-berlin.de}
\footnotetext[3]{Institut f\"ur Mathematik, MA 7-2, Technische
 Universit\"at Berlin, Str. des 17. Juni 136, 10623 Berlin, Germany. E-mail:
 suris@math.tu-berlin.de}

\begin{abstract}
We use the consistency approach to classify discrete integrable 3D
equations of the octahedron type. They are naturally treated on
the root lattice $Q(A_3)$ and are consistent on the
multidimensional lattice $Q(A_N)$. Our list includes the most
prominent representatives of this class, the discrete KP equation
and its Schwarzian (multi-ratio) version, as well as three further
equations. The combinatorics and geometry of the octahedron type
equations are explained. In particular the consistency on the
4-dimensional Delaunay cells has its origin in the classical
Desargues theorem of projective geometry. The main technical tool
used for the classification is the so called tripodal form of the
octahedron type equations.
\medskip

\medskip

\end{abstract}

\begin{quote} \small\tableofcontents\end{quote}

\section{Introduction}\label{s:intro}

This paper is devoted to the study of certain integrable discrete
three-dimensional equations. Prototypic examples of such equations
are the discrete KP (dKP), or the discrete bilinear Hirota equation, 
and its Schwarzian version. Both
equations play a pivotal role in modern mathematics \cite{FG_2006,
Henriques_2007, Kashaev_Reshetikhin_1997, Kazakov_2008,
King_Schief_2003, Konopelchenko_Schief_2002, KLWZ_1997,
Kuniba_Nakanishi_Suzuki_2010, LNQ_2009, NCW_1985, Speyer_2007,
Zabrodin_1997}.

The dKP equation for a complex lattice function (field) $x:\mathbb
Z^3\to\mathbb C$ was introduced in \cite{Hirota_1981} in the form
\begin{equation}\label{eq: Hirotas Hirota}
 \alpha x(m+e_1)x(m-e_1)+\beta x(m+e_2)x(m-e_2)+\gamma x(m+e_3)x(m-e_3)=0
\end{equation}
(here $m\in\mathbb Z^3$, $e_i$ is the unit vector of the $i$-th
coordinate direction, and $\alpha$, $\beta$, $\gamma$ are three
arbitrary complex coefficients). It relates six fields assigned to
any elementary octahedron of the even sublattice $\mathbb Z_{\rm
even}^3$ of $\mathbb Z^3$, centered at the odd point $m\in\mathbb
Z^3$. The even sublattice $\mathbb Z_{\rm even}^3$ is also known
as the face centered cubic lattice (fcc lattice). Its 3-cells are
octahedra and tetrahedra.

The fcc lattice is isomorphic to the root lattice of the type $A_3$,
\begin{equation}\label{eq: A3 lattice}
Q(A_3)=\big\{n=(n_0,\ldots,n_3)\in\mathbb Z^4:
n_0+n_1+n_2+n_3=0\big\}.
\end{equation}
The six vertices of every octahedron cell in $Q(A_3)$ can be written as $n+e_i+e_j$, $i,j\in\{0,1,2,3\}$ (with $n$ from the hyperplane $n_0+n_1+n_2+n_3=-2$). The dKP equation relates the fields $x$ at these vertices and reads:
\begin{eqnarray}\label{eq: root dKP}
 \alpha x(n+e_0+e_1)x(n+e_2+e_3)+\beta x(n+e_0+e_2)x(n+e_1+e_3)\nonumber\\
 +\gamma x(n+e_0+e_3)x(n+e_1+e_2) & = & 0.
\end{eqnarray}
This representation of the dKP equation turns out to be more
conceptual and plays the key role in the integrability analysis of
this paper.

Since on the hyperplane $n_0+n_1+n_2+n_3=0$ any coordinate is
determined by other three, by forgetting one coordinate we obtain
a one-to-one correspondence of the lattice $Q(A_3)$ with $\mathbb
Z^3$, $(n_0,n_1,n_2,n_3)\leftrightarrow(n_1,n_2,n_3)$. In this
realization, the dKP equation becomes an equation on $\mathbb Z^3$:
\begin{equation}\label{eq: Miwas dKP}
 \alpha x(n+e_1)x(n+e_2+e_3)+\beta x(n+e_2)x(n+e_3+e_1)+\gamma
 x(n+e_3)x(n+e_1+e_2)=0.
\end{equation}
In this form, introduced in \cite{Miwa_1982}, it relates fields
assigned to six out of eight vertices of any elementary cube of
$\mathbb Z^3$ (the fields at one pair of the opposite vertices,
$n$ and $n+e_1+e_2+e_3$, do not participate in the equation).

Throughout the present paper, we will use the following
abbreviation for lattice functions: $x$ for $x(n)$, $x_i$ for
$x(n+e_i)$, $x_{-i}$ for $x(n-e_i)$, $x_{ij}$ for $x(n+e_i+e_j)$, etc. 
In this notation,
the dKP equation in the $Q(A_3)$-form (\ref{eq: root dKP}) can be put as
\begin{equation}\label{eq: root dKP shifts}
 \alpha x_{01}x_{23}+\beta x_{02}x_{13}+\gamma x_{03}x_{12}=0,
\end{equation}
while in the $\mathbb Z^3$-form (\ref{eq: Miwas dKP}) it can be put as
\begin{equation}\label{eq: Miwas Hirota}
 \alpha x_1x_{23}+\beta x_2x_{13}+\gamma x_3x_{12}=0.
\end{equation}
We will call equations like \eqref{eq: Miwas Hirota} the {\em
octahedron type equations}, as opposed to the {\em cube type
equations}, whose best known representative is the dBKP equation
introduced in \cite{Miwa_1982}:
\begin{equation}\label{eq: Miwas Miwa}
 \alpha x_1x_{23}+\beta x_2x_{13}+\gamma x_3x_{12}+\delta xx_{123}=0.
\end{equation}
This latter equation relates the fields at all eight vertices of
any elementary cube of $\mathbb Z^3$. The octahedron type
equations could be considered as a subclass of the cube type
equations, however they have different properties and require for
a different analysis.
\begin{definition}\label{def: oct eq}
Equation of octahedron type is the relation
\[
 F(x_{01},x_{02},x_{03},x_{12},x_{13},x_{23})=0
\]
for the unknown function $x:Q(A_3)\to\mathbb C$, or, equivalently, the relation
\[
 F(x_1,x_2,x_3,x_{12},x_{13},x_{23})=0
\]
for the unknown function $x:\mathbb Z^3\to\mathbb C$.
\end{definition}

Our only assumption concerning the function $F$ is that it should
be locally analytic and satisfy
\begin{quote}
{\em irreducibility condition:} equation $F=0$ can be locally solved
with respect to any variable, the result depending on all other
variables, i.e., on any solution of $F=0$ we have $F_x\not\equiv 0$ , where $x$ denotes any argument of the equation.
\end{quote}
This condition forbids, in particular, equations with
$F$ of the form $F=AB$, where $A$ or $B$ does not depend on some
variable (for example, if $A_x\equiv 0$, then $F_x=AB_x$, and this
vanishes on the solution $A=0$). Also the so called ultradiscrete
equations with piecewise constant functions $F$ are excluded from
consideration. On the other hand, the irreducibility condition
does not forbid the case when the solution of $F=0$ is
multivalued, for instance, $F=AB$ with both $A$ and $B$ depending
on all variables. We handle such situations by working always in
the neighborhood of some solution branch, where the theorem on
implicit function applies.

It is important to observe that irreducible equations of the
octahedron type are 3D systems in the following sense: a generic
solution of such an equation on $\Integer^3$ can be defined by a
well posed Cauchy problem with generic 2D initial data. For
instance, for the dKP equation in the form \eqref{eq: Hirotas
Hirota} such initial data are constituted by prescribing the
values of $x$ on the planes $n_3=0$ and $n_3=-1$.

In the work \cite{Adler_Bobenko_Suris_2003} we pushed forward the
idea that integrability of discrete equations is synonymous with
their multidimensional consistency. For a certain class of 2D
equations, this notion was put in the basis of a classification of
integrable cases. In the present work, we classify
multidimensionally consistent (integrable) equations of the
octahedron type. A general idea of consistency leads to the following formulation.

\begin{definition}\label{def: consist cubic}
Consider a system on $\Integer^N$ consisting of (possibly different) octahedron type
equations
\begin{equation}\label{eq:oct gen cubic}
F(x_i,x_j,x_k,x_{ij},x_{ik},x_{jk})=0
\end{equation}
on all affine 3D sublattices $c+\Integer
e_i+\Integer e_j+\Integer e_k$. It is called multidimensionally
consistent if it has a solution whose restrictions on all 3D
sublattices are generic solutions of corresponding equations.
\end{definition}

This definition is a literal repetition of the corresponding notion for cubic type equations and does not take into account the specific feature of the octahedron type situation. If the lattice $\Integer^N$ lattice in this definition is treated as a realization of the root lattice
\begin{equation}\label{eq: AN lattice}
Q(A_N)=\big\{n=(n_0,\ldots,n_N)\in\mathbb Z^{N+1}: n_0+n_1+\ldots+n_N=0\big\},
\end{equation}
by forgetting the coordinate $n_0$, i.e.,
$Q(A_N)\ni(n_0,n_1,\ldots,n_N)\leftrightarrow(n_1,\ldots,n_N)\in\mathbb Z^N$, then the
corresponding octahedron type equations (\ref{eq:oct gen cubic}) live on the sublattices
\[
Q(A_3)=\{(n_0,n_i,n_j,n_k): n_0+n_i+n_j+n_k={\rm const}\}.
\]
However, these are by far not all $Q(A_3)$ sublattices of
$Q(A_N)$, but only those
involving the coordinate $n_0$, which
becomes distinguished in this formulation. A more symmetric notion
would be:
\begin{definition}\label{def: consist oct}
Consider a system on $Q(A_N)$ consisting of (possibly different) octahedron type equations
\begin{equation}\label{eq:oct gen}
F(x_{im},x_{jm},x_{km},x_{ij},x_{ik},x_{jk})=0
\end{equation}
on all affine sublattices $Q(A_3)=\{(n_i,n_j,n_k,n_m):
n_i+n_j+n_k+n_m={\rm const}\}$. It
is called multidimensionally
consistent if it has a solution whose restrictions on all $Q(A_3)$
sublattices are generic solutions of corresponding equations.
\end{definition}
This version is more natural but looks much more restrictive, since there are $N+1\choose 4$ sublattices $Q(A_3)$ of $Q(A_N)$, and only $N\choose 3$ sublattices $\mathbb Z^3$ of $\mathbb Z^N$. For instance, in the case $N=4$ the lattice $Q(A_4)$ has 5 sublattices $Q(A_3)$, of which only four involve the distinguished coordinate axis $n_0$. Nevertheless, in Section \ref{s:3} we prove the following result (see Proposition \ref{st: 3-->5}).
\begin{theorem}
Definitions \ref{def: consist cubic} and \ref{def: consist oct} of
the multidimensional
consistency of octahedron type equations are
equivalent.
\end{theorem}

Now we formulate the main result of this work. We classify
multidimensionally consistent
systems of octahedron type equations
(in notation of Definition \ref{def: consist oct}) modulo the
group of {\em admissible transformations}. This group consists of
changes of independent variables $n$ generated by the affine Weyl
group of $Q(A_N)$ (permutations of indices and translations by
lattice vectors) extended by the simultaneous inversion of all
coordinates $n\mapsto -n$, as well as of non-autonomous point
transformations of dependent variables, $x(n)\mapsto f(x(n),n)$.

\begin{theorem}\label{th:main}
Any multidimensionally consistent system of octahedron type equations on $Q(A_N)$ can be reduced by an admissible transformation to a system whose restriction to any sublattice $Q(A_3)$ has one of the following forms:
\begin{align}
\label{h1}\tag{\mbox{$\chi_1$}}
 & x_{ij}x_{km}-x_{ik}x_{jm}+x_{jk}x_{im}=0,\\[0.2em]
\label{h2}\tag{\mbox{$\chi_2$}}
 & \frac{(x_{im}-x_{ij})(x_{jm}-x_{jk})(x_{km}-x_{ik})}
        {(x_{ij}-x_{jm})(x_{jk}-x_{km})(x_{ik}-x_{im})}=-1,\\[0.4em]
\label{h3}\tag{\mbox{$\chi_3$}}
 & (x_{ik}-x_{ij})x_{im}+(x_{ij}-x_{jk})x_{jm}+(x_{jk}-x_{ik})x_{km}=0,\\[0.5em]
\label{h4}\tag{\mbox{$\chi_4$}}
 & \frac{x_{ik}-x_{ij}}{x_{im}}+\frac{x_{ij}-x_{jk}}{x_{jm}}
  +\frac{x_{jk}-x_{ik}}{x_{km}}=0,\\[0.5em]
\label{h5}\tag{\mbox{$\chi_5$}}
 & \frac{x_{ik}-x_{jk}}{x_{km}}=x_{ij}\Bigl(\frac{1}{x_{jm}}-\frac{1}{x_{im}}\Bigr).
\end{align}
Moreover, each of equations (\ref{h1}), (\ref{h2}) is
multidimensionally consistent with itself (that is, with the set
of the same equations for all $Q(A_3)$ sublattices), while
multidimensional sets involving any of the equations (\ref{h3}),
(\ref{h4}) include with necessity (\ref{h2}) on some of the
$Q(A_3)$ sublattices, and multidimensional sets involving
(\ref{h5}) include with necessity (\ref{h4}) on some of the
$Q(A_3)$ sublattices. The detailed description of consistent sets
is given in Theorem \ref{th:class}.
\end{theorem}

\begin{remark}
All these equations in $\mathbb Z^3$-form already appeared in the
literature. They were derived by the direct linearization method
in \cite{CWN_1986, NCWQ_1984, NCW_1985, Nijhoff_Capel_1990,
Dorfman_Nijhoff_1991, Bogdanov_Konopelchenko_1998}. First steps
towards consistency of the dKP equation on a root lattice of type
$A$ were made in \cite{Shinzawa_Saito_1998, Shinzawa_2000}.
\end{remark}
\medskip

\begin{remark}\label{rmk: Z^3 form}
It is instructive to look at the (somewhat more traditional)
$\mathbb Z^3$-form of equations from Theorem \ref{th:main}. Recall
that they are obtained by forgetting one of the indices. The high
symmetry grade of equations (\ref{h1}), (\ref{h2}) yields that
forgetting any one of the indices leads to the same equation on
$\mathbb Z^3$, namely
\[
x_{ij}x_k-x_{ik}x_j+x_{jk}x_i=0,
\]
resp.
\[
\frac{(x_{i}-x_{ij})(x_j-x_{jk})(x_k-x_{ik})}{(x_{ij}-x_{j})(x_{jk}-x_k)(x_{ik}-x_i)}=-1.
\]
As for equations (\ref{h3})--(\ref{h5}), one gets for each of them several
seemingly different equations on $\mathbb Z^3$, like
\[
\frac{x_{ik}-x_{ij}}{x_i}+\frac{x_{ij}-x_{jk}}{x_j}+\frac{x_{jk}-x_{ik}}{x_k}=0
\]
and
\[
\frac{x_{i}-x_{j}}{x_k}+\frac{x_{j}-x_{ij}}{x_{jk}}+\frac{x_{ij}-x_{i}}{x_{ik}}=0,
\]
both of which follow from (\ref{h4}) by forgetting one of the
indices. Without the unifying and symmetric $Q(A_3)$ notation, it
is not easy to recognize the equivalence of the latter two
equations. For instance, the (non-commutative versions of) these
equations are listed in \cite{Nijhoff_Capel_1990} as different
equations (1.6) and (1.9).
\end{remark}
\medskip

\begin{remark}\label{rmk: limit}
One can get equations (\ref{h4}), (\ref{h5}) from (\ref{h2}) via
simple limiting transitions. Performing in (\ref{h2}) a
non-autonomous point transformation $x(n)\mapsto
\epsilon^{-n_m}x(n)$, which amounts to the replacement of $x_{im},
x_{jm}, x_{km}$ by $\epsilon^{-1} x_{im}, \epsilon^{-1}
x_{jm},\epsilon^{-1} x_{km}$, and then sending $\epsilon\to 0$, we
arrive at equation (\ref{h4}). Analogously, performing in 
(\ref{h4}) the point transformation $x(n)\mapsto
\delta^{n_k}x(n)$, which amounts to the replacement of $x_{ik},
x_{jk}, x_{km}$ by $\delta x_{im}, \delta x_{jm}, \delta x_{km}$,
and then sending $\delta\to 0$, we arrive at equation (\ref{h5}).
Of course, these limiting transitions do not belong to our group
of admissible transformations.
\end{remark}
\medskip

The contents of the paper is as follows.

In Section \ref{s:Hirota} we give, for the sake of completeness, a
simple and well known derivation of the dKP equation and the
related octahedron type equations from the compatibility of
auxiliary linear problems. The main result of the present work,
Theorem \ref{th:main} (or its detailed version Theorem
\ref{th:class}), says that equations derived in this section
exhaust the list of multidimensionally consistent equations of
this type.

Section \ref{s:problem} is devoted to the definition of
multidimensional consistency for the octahedron type equations.
This definition is not quite straightforward, since the underlying
lattice $Q(A_3)$ is more sophisticated than the standard cubic one. 
The main problem is to
find a suitable multidimensional lattice containing 
several copies of the $Q(A_3)$ lattice which can simultaneously
support generic solutions of the discrete octahedron type
equations. The fact that this problem is not trivial is
illustrated in Section \ref{ss:no fcc} by a failure of one
possible definition. A successful definition is then illustrated
by the example of the dKP equation in Section \ref{subsect: dKP
consist} and formulated in full generality in Section \ref{s:3}.
An elementary combinatorial cell (a 4-cell of the root lattice $Q(A_4)$ with five octahedral faces) is best illustrated by the Desargues configuration, see Section
\ref{s: cells}.

Section \ref{s:var3} contains the central technical observation:
each octahedron equation of the multidimensionally consistent
system can be written in eight ways in the so called tripodal
form, and the tripodal forms of equations on the adjacent
octahedra must combine themselves in a very special way. Actually,
the necessary conditions for consistency established in this
section are the basis for the subsequent solution of the
classification problem.

This solution starts in Section \ref{s:3leg} where we classify all
octahedron type equations admitting eight tripodal forms. This
conditions turns out to be stringent enough to produce a finite
list of equations, given in Theorem \ref{th:3leg_types}.

Finally, in Section \ref{s:class} we combine tripodal equations into
consistent systems, a complete list of which is given in Theorem \ref{th:class}.

\section{Equations of octahedron type through linear problems}
\label{s:Hirota}

All equations from the list are related to each other via
difference substitutions, and can be easily derived from simple
linear problems like
\begin{equation}\label{lin.lin}
  f_2-f_1=af,\quad f_3-f_1=bf.
\end{equation}
It should be noticed that all four faces of the tetrahedron with
the vertices $f,f_1,f_2,f_3$ are on the same footing, because, due
to (\ref{lin.lin}), there hold also the further linear equations:
\begin{equation}\label{lin.lin1}
  f_3-f_2=cf,\quad f_2-f_1=d(f_3-f_1),
\end{equation}
where $c=b-a$, $d=a/b$. Any two out of the four equations
(\ref{lin.lin}), (\ref{lin.lin1}) are equivalent to
(\ref{lin.lin}).

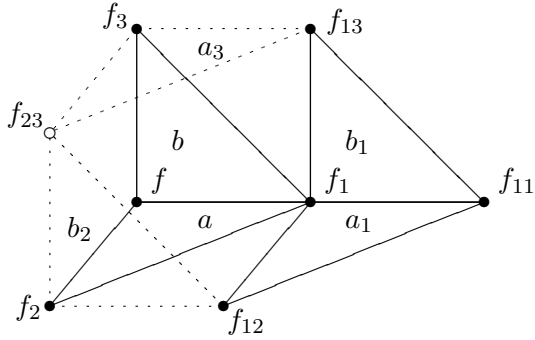
\begin{figure}[htbp]
\begin{center}
\setlength{\unitlength}{0.06em}
\begin{picture}(400,180)(-120,-60)
 \drawline(0,0)(100,0)(-50,-60)(0,0)                   
 \drawline(0,0)(100,0)(0,100)(0,0)                     
 \dashline{2}(-47.5,42.5)(0,100)(100,100)(-47,41.5)    
 \dashline{2}(-50,37)(-50,-60)(50,-60)(-47.5,37.5)     
 \drawline(100,0)(200,0)(50,-60)(100,0)                
 \drawline(100,0)(200,0)(100,100)(100,0)               
 \put(0,100){\circle*{6}}
 \put(100,100){\circle*{6}}
 \put(-50,-60){\circle*{6}}
 \put(50,-60){\circle*{6}}
 \put(0,0){\circle*{6}}
 \put(100,0){\circle*{6}}
 \put(200,0){\circle*{6}}
 \put(-50,40){\circle{6}}
 \put(35,-15){$a$}
 \put(20,30){$b$}
 \put(120,-15){$a_1$}
 \put(120,30){$b_1$}
 \put(35,85){$a_3$}
 \put(-40,-20){$b_2$}
 \put(8,8){$f$}
 \put(108,8){$f_1$}
 \put(208,8){$f_{11}$}
 \put(-70,-65){$f_2$}
 \put(52,-73){$f_{12}$}
 \put(-20,103){$f_3$}
 \put(109,100){$f_{13}$}
 \put(-75,46){$f_{23}$}
\end{picture}
\end{center}
\caption{Linear equations on triangles.}\label{cold:hex.fig1}
\end{figure}
It is easy to compute the compatibility conditions (which express
the equality of two alternative expressions for $f_{23}$ through
the initial values $f$, $f_1$, $f_{11}$, see Figure
\ref{cold:hex.fig1}):
\begin{equation}\label{lin.ab}
  a_3+b_1=b_2+a_1,\quad a_3b=b_2a.
\end{equation}
These formulas should be interpreted as a map
$(a,a_1,b,b_1)\mapsto(a_3,b_2)$. This map defines the evolution of
the initial data consisting, for instance, of the values of $a$
prescribed on the coordinate plane $12$ and the values of $b$
prescribed on the coordinate plane $13$. Expressing in
(\ref{lin.ab}) $a$ and $b$ through $f$ according to
(\ref{lin.lin}), we obtain an equation connecting 6 out of 8
values of $f$ at the vertices of an elementary cube, of the type
(\ref{h4}):
\begin{equation}\label{lin.f}
  \frac{f_{13}-f_{12}}{f_1}
 +\frac{f_{12}-f_{23}}{f_2}
 +\frac{f_{23}-f_{13}}{f_3}=0.
\end{equation}
Thus, there naturally appears the decomposition of the cubic
lattice into tetrahedra (which support the linear problem) and
octahedra (which support the nonlinear equations).

Each equations in (\ref{lin.ab}) can be interpreted as a
conservation law and allows for an introduction of a potential.
For instance, due to the first equation in (\ref{lin.ab}) one can
introduce a function $\rho$ on the vertices of the cubic lattice,
such that
\[
 a=\rho_2-\rho_1,\quad b=\rho_3-\rho_1,
\]
which brings the second equation in (\ref{lin.ab}) into the form
\begin{equation}\label{lin.rho}
 (\rho_{23}-\rho_{13})(\rho_3-\rho_1)=(\rho_{23}-\rho_{12})(\rho_2-\rho_1),
\end{equation}
which is equivalent to (\ref{h3}). Analogously, the second
equation allows us to introduce a vertex function $q$ such that
\[
 a=\frac{q_2}{q},\quad b=\frac{q_3}{q},
\]
which brings the first equation to the form
\begin{equation}\label{lin.q}
 \frac{q_{12}-q_{13}}{q_1}=q_{23}\Bigl(\frac1{q_3}-\frac1{q_2}\Bigr).
\end{equation}
which is nothing but (\ref{h5}). Finally, setting
\[
 q=\frac{\tau_1}{\tau}\quad\Rightarrow\quad
 a=\frac{\tau_{12}\tau}{\tau_1\tau_2},\quad
 b=\frac{\tau_{13}\tau}{\tau_1\tau_3},
\]
we can rewrite equation (\ref{lin.q}) as
\[
(T_1-1)\left(\frac{\tau_2\tau_{13}}{\tau_1\tau_{23}}
 -\frac{\tau_3\tau_{12}}{\tau_1\tau_{23}}\right)=0,
\]
which yields the Hirota equation (\ref{h1})
\begin{equation}\label{lin.tau}
 \tau_2\tau_{13}-\tau_3\tau_{12}=c\tau_1\tau_{23}.
\end{equation}

In order to obtain equation (\ref{h2}), one has to start with the
linear problems
\begin{equation}\label{lin.lin'}
  \psi_2-\psi=u(\psi_1-\psi),\quad \psi_3-\psi=v(\psi_1-\psi).
\end{equation}
These linear problems are not essentially different from
(\ref{lin.lin}): on one hand, they are gauge equivalent, and on
the other hand, they have already appeared in (\ref{lin.lin1}), so
that on lattices of a sufficiently large dimension both types of
linear problems coexist. The condition of compatibility of
(\ref{lin.lin'}) leads to the nonlinear equations
\begin{equation}\label{lin.uv}
 (u_3-1)(v-1)=(v_2-1)(u-1),\quad u_3v_1=v_2u_1,
\end{equation}
which are interpreted, as before, as a map
$(u,u_1,v,v_1)\mapsto(u_3,v_2)$. Variables $\psi$ satisfy equation
of the type (\ref{h2}):
\begin{equation}\label{lin.psi}
 \frac{(\psi_1-\psi_{12})(\psi_2-\psi_{23})(\psi_3-\psi_{13})}
      {(\psi_{12}-\psi_2)(\psi_{23}-\psi_3)(\psi_{13}-\psi_1)}=-1.
\end{equation}
Like in the previous example, both equations (\ref{lin.uv}) can be
interpreted as conservation laws and allow for an introduction of
a potential. However, this time these two equations are similar
and lead both to the same equations. For instance, one can resolve
the second equation in (\ref{lin.uv}) by introducing a vertex
function $f$ according to
\[
  u=\frac{f_1}{f_2},\quad v=\frac{f_1}{f_3},
\]
and then the first equation in (\ref{lin.uv}) leads again to
(\ref{lin.f}).

One can add one more equation which appears via the substitution
\[
  h=\frac{q_1}{q_3}=\frac{\tau_{12}\tau_3}{\tau_1\tau_{23}},
\]
which leads to
\begin{equation}\label{lin.Y}
 (h_{12}-1)(h_3-1)=h_2h_{13}(1-h^{-1}_1)(1-h^{-1}_{23}).
\end{equation}
This equation will be discussed in Section \ref{sect: Y-eq} under
the name (\ref{h6}). One can express the variables $h$ directly
through $\psi$, by composing all intermediate transformations:
\[
 h=\frac{q_1}{q_3}=\frac{a}{b}=\frac{f_1-f_2}{f_3-f_2}=\frac{1/u-1}{1/v-1}
  =\frac{\dfrac{\psi_2-\psi}{\psi_1-\psi}-1}{\dfrac{\psi_2-\psi}{\psi_3-\psi}-1}
  =\frac{(\psi_2-\psi_1)(\psi_3-\psi)}{(\psi_2-\psi_3)(\psi_1-\psi)}.
\]
It can be checked directly that this substitution turns equation
(\ref{lin.Y}) into a product of four copies of equation
(\ref{lin.psi}), at the original vertex and at its shifts in three
coordinate directions.

Thus, all equations (\ref{h1})--(\ref{h5}) and (\ref{lin.Y}) are
related to one another in a rather simple manner:
\[
\begin{array}{ccccccccccc}
&&&&&& (\ref{lin.rho})=(\ref{h5}) && (\ref{lin.Y})=(\ref{h6}) &  &\\
  &&&&&& \rho & & h & &\\
  &&&&&& \downarrow && \uparrow &&\\
 \psi &\to & u,v &\leftarrow & f &\to & a,b &\leftarrow & q &\leftarrow & \tau\\[0.5em]
  (\ref{lin.psi})=(\ref{h2}) &&&& (\ref{lin.f})=(\ref{h4})
 &&&&(\ref{lin.q})=(\ref{h5}) && (\ref{lin.tau})=(\ref{h1})
\end{array}
\]


\section{Multidimensional consistency}\label{s:problem}

\subsection{Consistency for equations of the cube type}
\label{ss:cube}

The idea behind the definition of multidimensional consistency given in Section \ref{s:intro} (Definitions \ref{def: consist cubic}, \ref{def: consist oct}) is rather general and can be implemented for various types of discrete systems. However, it refers to the properties of equations on the whole multidimensional lattice, which is difficult to verify. It would be preferable to have some sufficient conditions for multidimensional consistency which refer to some local (finite) piece of the lattice.

For 3D equations of the cube type, like dBKP equation \eqref{eq: Miwas Miwa}, such a local sufficient
condition refers to one 4D cube, as follows (see, e.g., \cite{Adler_Bobenko_Suris_2003}). A solution of a
non-degenerate equation of the cube type which is generic on each 3D
sublattice can be defined by the initial data consisting of the
values of $x$ on all two-dimensional coordinate planes. Indeed,
these data are obviously independent, and the inductive
application of the equation allows one to extend the solution from
the coordinate planes to the whole of $\mathbb Z^N$ (compute $x_{ijk}$ from the known $x$, $x_i$ and $x_{ij}$), provided one
does not encounter contradictions in this inductive process. It is
sufficient to verify the lack of contradictions within one 4D
cube, which is done as follows. Initial data within one 4D cube
are:
\[
x, \; x_i, \; x_{ij} \; (1\le i<j\le 4).
\]
Application of equations on the four cubic faces adjacent to the
vertex $x$ yields the values $x_{ijk}\;$ $(1\le i<j<k\le 4)$, and
then application of equations on the four cubic faces adjacent to
the vertex $x_{1234}$ yields four {\em a priori} different values
for this last field. The 4D consistency takes place (and implies
multidimensional consistency) if these four values identically
coincide (thus, one has three conditions in terms of 11 initial
data to be fulfilled). A more detailed discussion can be found in
\cite{Adler_Bobenko_Suris_2003}, where it was shown that the dBKP
equation satisfies this criterium.

\subsection{Lack of consistency for the dKP equation on the face
centered cubic lattice} \label{ss:no fcc}

The type of initial value problem discussed for the cube type
equations in the previous section is not applicable to 3D
equations of the octahedron type, because of the absence of the
fields $x_{ijk}$ from the equations (\ref{eq:oct gen}). Hence, the
very notion of the multidimensional consistency in this concrete
situation has to be modified.

Here we analyze a possible definition of 4D consistency for the dKP equation suggested by its original form (\ref{eq: Hirotas Hirota}). This definition turns out to be unsuccessful, so that this section is not necessary for the further reading, however we hope that it will clearly demonstrate the non-triviality of the problem of finding the suitable notion.

Equation (\ref{eq: Hirotas Hirota}) decomposes into two independent systems, one for the fields $x:\mathbb Z_{\rm even}^3\to \mathbb C$ on the so called even (or black) sublattice
\[
\mathbb Z_{\rm even}^3=\left\{m=(m_1,m_2,m_3)\in\mathbb Z ^3:
m_1+m_2+m_3\equiv 0\pmod 2\right\},
\]
and another one for the fields $x:\mathbb Z_{\rm odd}^3\to \mathbb
C$ on the odd (or white) sublattice
\[
\mathbb Z_{\rm odd}^3=\left\{m=(m_1,m_2,m_3)\in\mathbb Z ^3:
m_1+m_2+m_3\equiv 1\pmod 2\right\};
\]
it will be enough to consider the half defined on $\mathbb
Z^3_{\rm even}$ (say). The latter is known as the face centered cubic (fcc)
lattice. Its set of vertices is clearly in a
one-to-one correspondence to $\mathbb Z^3$, but its (Delaunay) cell structure is quite different.
Its 2-cells are equilateral triangles, while its 3-cells are octahedra and tetrahedra.
Equation (\ref{eq: Hirotas Hirota}) relates six fields
assigned to any elementary octahedron of $\mathbb Z^3_{\rm even}$,
centered at $m\in\mathbb Z_{\rm odd}^3$. Two octahedra with a nonempty
intersection either share an edge (those whose centers are
neighbors in $\mathbb Z_{\rm odd}^3$) or either share a vertex
(those whose centers are diagonal neighbors in $\mathbb Z_{\rm
odd}^3$, i.e., are at distance 2 from one another).

One could attempt to define the multidimensional consistency of
octahedron type equations by imposing them on all $\mathbb Z_{\rm
even}^3$ sublattices in
\[
\mathbb Z_{\rm even}^N=\left\{m=(m_1,\ldots,m_N)\in\mathbb Z^N:
m_1+\ldots+m_N\equiv 0\pmod 2\right\}.
\]
However, this idea turns out to be invalid. Indeed, one can show
the inconsistency of two copies of (\ref{eq: Hirotas Hirota}),
corresponding to the sublattices (123) and (124):
\begin{align}
\label{Hir11}
 x(m+e_3)x(m-e_3)-ax(m+e_1)x(m-e_1)+bx(m+e_2)x(m-e_2)&=0,\\
\label{Hir12}
 x(m+e_4)x(m-e_4)-cx(m+e_1)x(m-e_1)+dx(m+e_2)x(m-e_2)&=0,
\end{align}
where $a,b,c,d$ are arbitrary constants. As initial data for these
two equations one can take the values of $x$ at four
two-dimensional planes parallel to the coordinate plane 12:
\[
  x(m_1,m_2,0,0),\quad x(m_1,m_2,-1,0),\quad x(m_1,m_2,0,-1),\quad x(m_1,m_2,-1,-1).
\]
These data are free, in the sense that they are not subject to any
equation among (\ref{Hir11}), (\ref{Hir12}), or remaining two
equations corresponding to the sublattices 134 and 234, since
any of these equations contains at least one pair of points
differing by $2\ee_3$ or by $2\ee_4$. In particular, consider the
following points in these two planes:
\begin{gather*}
 (m_1,m_2)=(\pm2,0),~ (1,\pm1),~(0,0),~(-1,\pm1),~(0,\pm2)
 \quad\text{with}\quad (m_3,m_4)=(0,0);\\
 (m_1,m_2)=(\pm1,0),~ (0,\pm1)\quad\text{with}\quad (m_3,m_4)=(-1,0),~(0,-1);\\
 (m_1,m_2)=(0,0) \quad\text{with}\quad (m_3,m_4)=(-1,-1).
\end{gather*}
These are round points on Fig.~\ref{Fig: failure}. Now equations
(\ref{Hir11}), (\ref{Hir12}) allow us to compute $x$ at the triangular
points
\[
 (0,0,1,-1),\quad (0,0,-1,1),\quad
 (\pm1,0,1,0),\quad (0,\pm1,1,0),\quad
 (\pm1,0,0,1),\quad (0,\pm1,0,1),
\]
and then two different values at the square point $(0,0,1,1)$. It can
be checked that, for any choice of the non-vanishing coefficients
$a,b,c,d$, these two values of $x(0,0,1,1)$ do not coincide
identically (as functions of the initial data).

\begin{figure}[ht]
\centerline{\includegraphics[width=7cm]{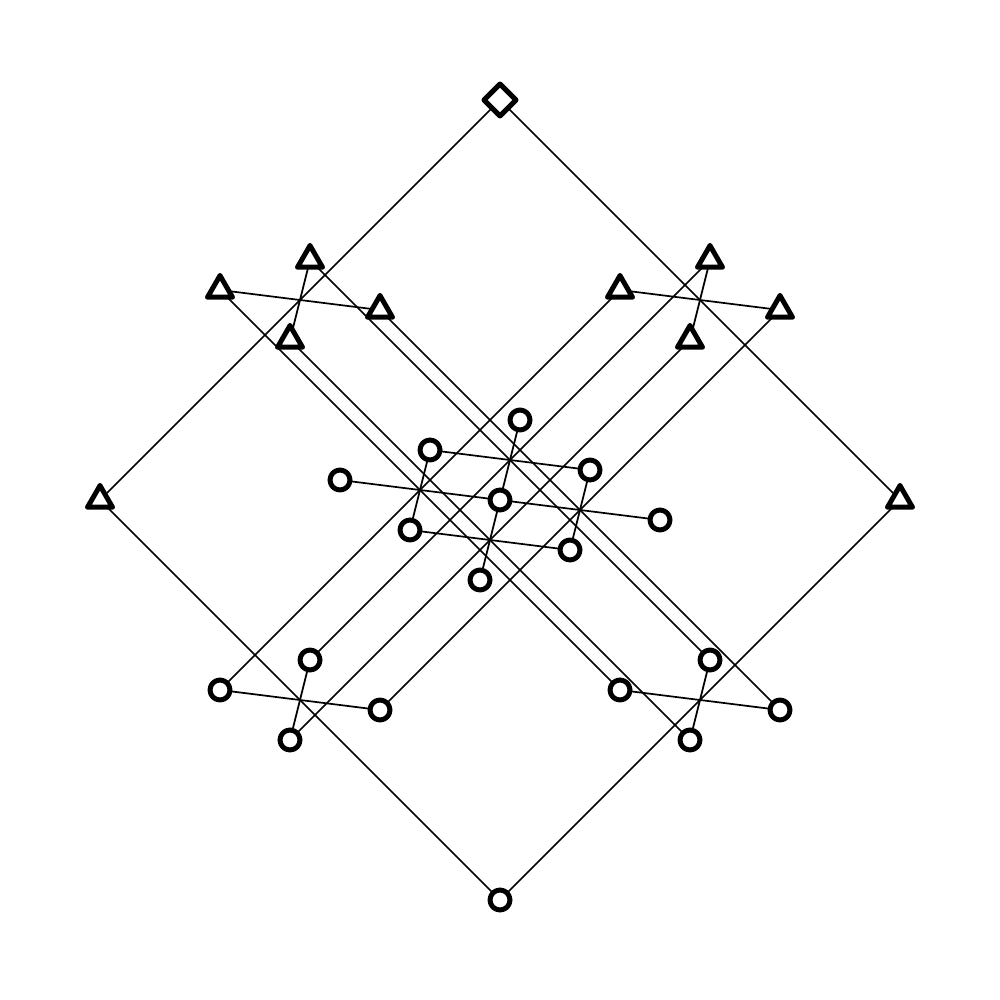}}
\caption{Consistency check for equations (\ref{Hir11}),
(\ref{Hir12}) leads to a contradiction at the uppermost vertex.}
\label{Fig: failure}
\end{figure}

Note that two neighboring octahedra in different three-dimensional
lattices $\mathbb Z^3$ with two common coordinate directions (as
those underlying equations (\ref{Hir11}), (\ref{Hir12})) do not share
triangular faces; rather, they share two pairs of antipodal
points, or an 2D equatorial square-formed section spanned by these points.

\subsection{Consistency of the dKP equation on a 4D cube}
\label{subsect: dKP consist}

Here we prove a local statement about the 4D consistency for the $\mathbb Z^3$ version (\ref{eq: Miwas Hirota}) of the dKP equation, which deals with the octahedra contained within one elementary 4D cube of $\mathbb Z^4$. Each of the eight 3D cubic faces of the 4D cube contains such an octahedron, so that we have the eight octahedra
\[
[jkm] =\{x_j,x_k,x_m,x_{jk},x_{jm},x_{km}\}\quad{\rm and}\quad T_i[jkm]=\{x_{ij},x_{ik},x_{im},x_{ijk},x_{ijm},x_{ikm}\}
\]
(where $\{i,j,k,m\}=\{1,2,3,4\}$). Neither of the equations on these octahedra involves the fields $x$ and $x_{1234}$, so that only 14 out of 16 vertices are involved. It turns out that one can take 9 of them as independent initial data, for instance,
\[
x_1,\; x_2,\; x_3,\; x_4,\; x_{14},\; x_{24},\; x_{34},\; x_{134},\; x_{234},
\]
and to find the remaining 5 fields using {\em only three of the equations}, namely those for the octahedra $[124]$, $[134]$, $[234]$:
\begin{equation}\label{dKP xij}
\begin{aligned}
 & x_{12}x_4-x_{14}x_2+x_{24}x_1=0,\\
 & x_{13}x_4-x_{14}x_3+x_{34}x_1=0,\\
 & x_{23}x_4-x_{24}x_3+x_{34}x_2=0,
\end{aligned}
\end{equation}
and for their shifted copies $T_3[124]$, $T_2[134]$, $T_1[234]$:
\begin{equation}\label{dKP xijk}
\begin{aligned}
 & x_{123}x_{34}-x_{134}x_{23}+x_{234}x_{13}=0,\\
 & x_{123}x_{24}-x_{124}x_{23}+x_{234}x_{12}=0,\\
 & x_{123}x_{14}-x_{124}x_{13}+x_{134}x_{12}=0.
\end{aligned}
\end{equation}
We use equations (\ref{dKP xij}) to determine $x_{12}$, $x_{13}$, $x_{23}$. Then, we use the first equation in (\ref{dKP xijk}) to determine $x_{123}$. Finally, we have two alternative answers for $x_{124}$ which come from the last two equations in (\ref{dKP xijk}). We now show that these two values coincide as functions of the initial data, so that the 4D consistency takes place. The following computations might make an impression of a repeated application of a certain skillful trick, but we will show later that this trick is actually a key structural feature common for all consistent octahedron type equations, namely the so called {\em tripodal form} of these equations.

First of all, we show that the values of $x_{12}$, $x_{13}$, $x_{23}$ determined from
(\ref{dKP xij}), together with the initial data $x_1$, $x_2$, $x_3$, automatically satisfy the dKP equation on the octahedron $[123]$. To this end, we rewrite (\ref{dKP xij}) as
\begin{equation}\label{dKP xij tripod 1}
\begin{aligned}
 & \frac{x_{12}}{x_1x_2}-\frac{x_{14}}{x_1x_4}+\frac{x_{24}}{x_2x_4}=0,\\
 & \frac{x_{13}}{x_1x_3}-\frac{x_{14}}{x_1x_4}+\frac{x_{34}}{x_3x_4}=0,\\
 & \frac{x_{23}}{x_2x_3}-\frac{x_{24}}{x_2x_4}+\frac{x_{34}}{x_3x_4}=0.
\end{aligned}
\end{equation}
An obvious linear combination of these equations immediately leads to
\[
\frac{x_{12}}{x_1x_2}-\frac{x_{13}}{x_1x_3}+\frac{x_{23}}{x_2x_3}=0,
\]
which is equivalent to
\begin{equation}\label{dKP 123}
x_{12}x_3-x_{13}x_2+x_{23}x_1=0.
\end{equation}

Second, we show that the values of $x_{12}$, $x_{13}$, $x_{23}$ determined from
(\ref{dKP xij}), together with $x_{14}$, $x_{24}$, $x_{34}$, automatically  satisfy an equation which literally coincides with dKP. For this aim, we rewrite equations (\ref{dKP xij}) in another equivalent form:
\begin{equation}\label{dKP xij tripod 2}
\begin{aligned}
& \frac{x_{12}}{x_{14}x_{24}}-\frac{x_2}{x_4x_{24}}+\frac{x_1}{x_4x_{14}}=0,\\
& \frac{x_{13}}{x_{14}x_{34}}-\frac{x_3}{x_4x_{34}}+\frac{x_1}{x_4x_{14}}=0,\\
& \frac{x_{23}}{x_{24}x_{34}}-\frac{x_3}{x_4x_{34}}+\frac{x_2}{x_4x_{24}}=0.
\end{aligned}
\end{equation}
A suitable linear combination of these equation leads to
\[
\frac{x_{12}}{x_{14}x_{24}}-\frac{x_{13}}{x_{14}x_{34}}+\frac{x_{23}}{x_{24}x_{34}}=0,
\]
which is equivalent to
\begin{equation}\label{dKP even}
x_{12}x_{34}-x_{13}x_{24}+x_{23}x_{14}=0.
\end{equation}

Finally, we show the 4D consistency. The first equation in (\ref{dKP xijk}) used to determine $x_{123}$ is equivalent to
\begin{equation}\label{dKP xijk tripod 1}
\frac{x_{34}}{x_{13}x_{23}}-\frac{x_{134}}{x_{13}x_{123}}+
\frac{x_{234}}{x_{23}x_{123}}=0.
\end{equation}
The last two equations in (\ref{dKP xijk}) used to determine $x_{124}$ are equivalent to
\begin{equation}\label{dKP xijk tripod 2}
\begin{aligned}
 & \frac{x_{24}}{x_{12}x_{23}}-\frac{x_{124}}{x_{12}x_{123}}+
 \frac{x_{234}}{x_{12}x_{123}}=0,\\
 & \frac{x_{14}}{x_{12}x_{13}}-\frac{x_{124}}{x_{12}x_{123}}+
 \frac{x_{134}}{x_{13}x_{123}}=0.
\end{aligned}
\end{equation}
They give the same value of $x_{124}$ if and only if
\[
\frac{x_{14}}{x_{12}x_{13}}-\frac{x_{24}}{x_{12}x_{23}}
+\frac{x_{134}}{x_{13}x_{123}}-\frac{x_{234}}{x_{23}x_{123}}=0.
\]
Combining the latter equation with (\ref{dKP xijk tripod 1}), we see that the 4D consistency condition is reduced to
\[
\frac{x_{14}}{x_{12}x_{13}}-\frac{x_{24}}{x_{12}x_{23}}+\frac{x_{34}}{x_{13}x_{23}}=0,
\]
which is equivalent to the already proven equation (\ref{dKP even}).

\subsection{Cell structure of the lattice $Q(A_N)$}
\label{s: cells}

Considerations of Section \ref{subsect: dKP consist} give a hint
towards a valid formulation of the notion of multidimensional
consistency for the octahedron type equations. In particular, equation (\ref{dKP even}) for the even shifts calls for an interpretation as an equation on a further octahedron present in the 4D lattice. Such interpretation is enabled by an embedding of the lattice $Q(A_3)$ into the root lattice $Q(A_N)$, given in (\ref{eq: AN lattice}). We start with a short description of the Delaunay cell structure of the four-dimensional lattice
$Q(A_4)$ \cite{Conway_Sloane_1991, Moody_Patera_1992}. We do not
go into detail and give an elementary description appropriate for
our purposes. For each $N$ there are $N$ sorts of $N$-cells of
$Q(A_N)$ denoted by $P(k,N)$, $k=1,\ldots,N$.

\begin{quote}
\noindent
{\em Two sorts of 2-cells:}
\begin{itemize}
\item[$P(1,2)$:] black triangles $\{x_i,x_j,x_k\}$,
\item[$P(2,2)$:] white triangles $\{x_{ij},x_{ik},x_{jk}\}$;
\end{itemize}
\end{quote}

\begin{quote}
\noindent
{\em Three sorts of 3-cells:}
\begin{itemize}
\item[$P(1,3)$:] black tetrahedra $\{x_i,x_j,x_k,x_{\ell}\}$, with all
four facets being black triangles,
\item[$P(2,3)$:] octahedra
$[ijk\ell]=\{x_{ij},x_{ik},x_{i\ell},x_{jk},x_{j\ell},x_{k\ell}\}$,
the eight triangular facets of each octahedron being bi-colored,
consult Fig.~\ref{fig: octahedron};
\item[$P(3,3)$:] white tetrahedra $\{x_{ijk},x_{ij\ell},x_{ik\ell},x_{jk\ell}\}$,
with all four facets being white triangles;
\end{itemize}
\end{quote}

\begin{quote}
\noindent
{\em Four sorts of 4-cells:}
\begin{itemize}
\item[$P(1,4)$:] black 4-simplices $\{x_i,x_j,x_k,x_{\ell},x_m\}$,
with all five facets being black tetrahedra,
\item[$P(2,4)$:] 4-ambo-simplices, in the terminology of \cite{Conway_Sloane_1991},
with the ten vertices $x_{ij}$, $i,j\in\{0,1,\ldots,4\}$; each such polytope
has five octahedral facets $\langle i\rangle=[jk\ell m]$ and five black tetrahedral
facets $T_i\{x_j,x_k,x_\ell,x_m\}$;
\item[$P(3,4)$:] 4-ambo-simplices with the ten vertices
$x_{ijk}$, $i,j,k\in\{0,1,2,3,4\}$; each of these polytopes has five octahedral
facets $T_i[jk\ell m]$ and five white tetrahedral facets;
\item[$P(4,4)$:] white 4-simplices
$\{x_{ijk\ell},x_{ijkm},x_{ij\ell m},x_{ik\ell m},x_{jk\ell m}\}$,
with all five facets being white tetrahedra.
\end{itemize}
\end{quote}

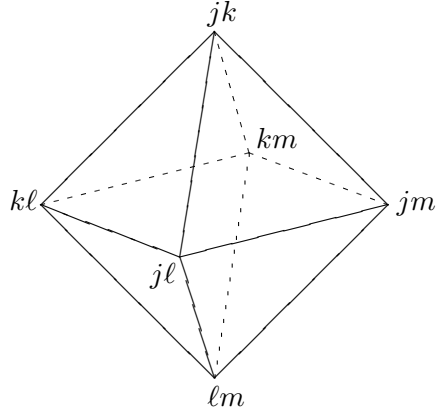
\begin{figure}[thb]
\setlength{\unitlength}{0.06em}
\centerline{
\begin{picture}(200,240)(-100,-120)
 \path(-100,0)(0,100)(100,0)(0,-100)(-100,0)(-20,-30)(100,0)
 \path(0,100)(-20,-30)(0,-100)
 \dashline{3}(0,100)(20,30)(0,-100)
 \dashline{3}(-100,0)(20,30)(100,0)
 \put(-4,106){$jk$}\put(-4,-116){$\ell m$}
 \put(-118,-3){$k\ell$}\put(105,-3){$jm$}
 \put(-37,-45){$j\ell$}\put(24,33){$km$}
\end{picture}}
\caption{Vertex enumeration of the octahedron $\langle i\rangle=[jk\ell m]$, where $i\in\{0,1,2,3,4\}$ and $\{j,k,\ell,m\}=\{0,1,2,3,4\}\setminus\{i\}$. Opposite vertices carry complementary pairs of indices. Facets are bi-colored; there are four white triangles like $\{jk,j\ell,k\ell\}$ missing one of the indices ($m$ in this case), and four black triangles like $\{jk,j\ell,jm\}=T_j\{k,\ell,m\}$ sharing one common index ($j$ in this case).}
\label{fig: octahedron}
\end{figure}

The combinatorial description of higher dimensional cells in the
root lattices
$Q(A_N)$ with bigger $N$ is given analogously.

The affine Weyl group is generated by permutations of indices and
translations. It permutes the cells of the same sort.

We finish this section with a suggestive description of the
admittedly somewhat complicated combinatorics of a 4-ambo-simplex
$P(2,4)$ whose five octahedral faces $\EQ0$, $\EQ1$, $\EQ2$,
$\EQ3$, $\EQ4$ carry a quintuple of consistent octahedron type
equations. Remarkably, this description not only has a
combinatorial meaning, but also has a direct relation to the
multidimensional consistency of equation (\ref{h2}), explained in
\cite[p. 285]{DDG_book}. Consider a map $x:Q(A_N)\to{\mathbb
R}P^n$ satisfying the following condition: the image of any white
triangle $\{x_{ij},x_{ik},x_{jk}\}$ is a collinear triple of
points. Such maps were introduced in the three-dimensional
situation by Schief \cite{Schief_talk} under the name
``Laplace-Darboux lattices''. He also observed the relation of
their four-dimensional consistency to the Desargues theorem
(private communication; see \cite{Adler_talk, Bobenko_talk}).
These maps (under the name ``Desargues maps'') are studied in
detail in the recent work by Doliwa \cite{Doliwa_2010,
Doliwa_prepr}. A connection of the Desargues theorem to equation
(\ref{h4}) appeared in \cite{Adler_tangent}.

It is easy to realize that the image of an octahedron $[ijkm]$ is
a complete quadrilateral, as on Fig.~\ref{fig:Menelaus}. It
contains four lines which are images of the white triangular faces
$\{x_{ij},x_{ik},x_{jk}\}$, and four triangles which are images of
the black triangular faces $T_i\{x_j,x_k,x_m\}$. An analytic
description of the complete quadrilateral is given, according to
the classical Menelaus theorem, by equation (\ref{h2}).

\begin{figure}[thb]
\centerline{\includegraphics[width=40mm]{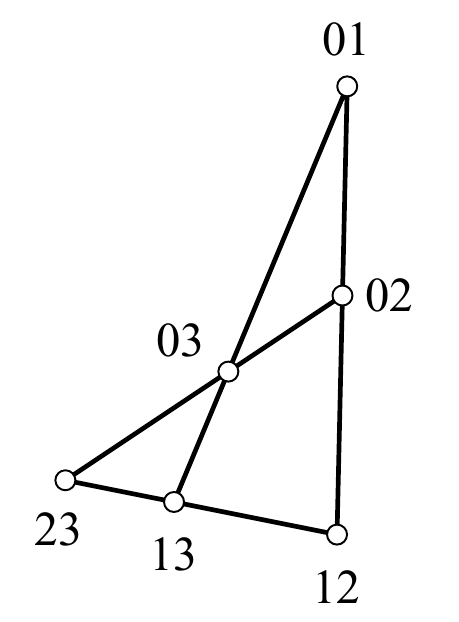}}
\caption{Geometric interpretation of equation (\ref{h2}): a
Menelaus configuration, or a complete quadrilateral, is an image
of the octahedron $[0123]$; four lines correspond to the white
triangles $\{jk,j\ell,k\ell\}$, while four triangles correspond to
the black triangles $\{jk,j\ell,jm\}$.} \label{fig:Menelaus}
\end{figure}

The image of the 4-ambo-simplex $P(2,4)$ is then a configuration
like the one on Fig.~\ref{fig:Desargues}. One can recognize here
five complete quadrilaterals, corresponding to the octahedra
$\langle i\rangle=[jk\ell m]$, as well as the images of the five
black tetrahedra $T_i\{x_j,x_k,x_\ell,x_m\}$. The ten lines are
the images of the ten white triangular faces of the octahedra. The
configuration on Fig.~\ref{fig:Desargues} illustrates one of the
most important incidence theorems of the classical projective
geometry -- the Desargues theorem. Its five-fold symmetry is not
obvious from the first glance (and from the original formulation
of the theorem), but it was well known to the classics, see, e.g.,
\cite{Hilbert-Cohn-Vossen}.

\begin{figure}[thb]
\centerline{\includegraphics[width=8cm]{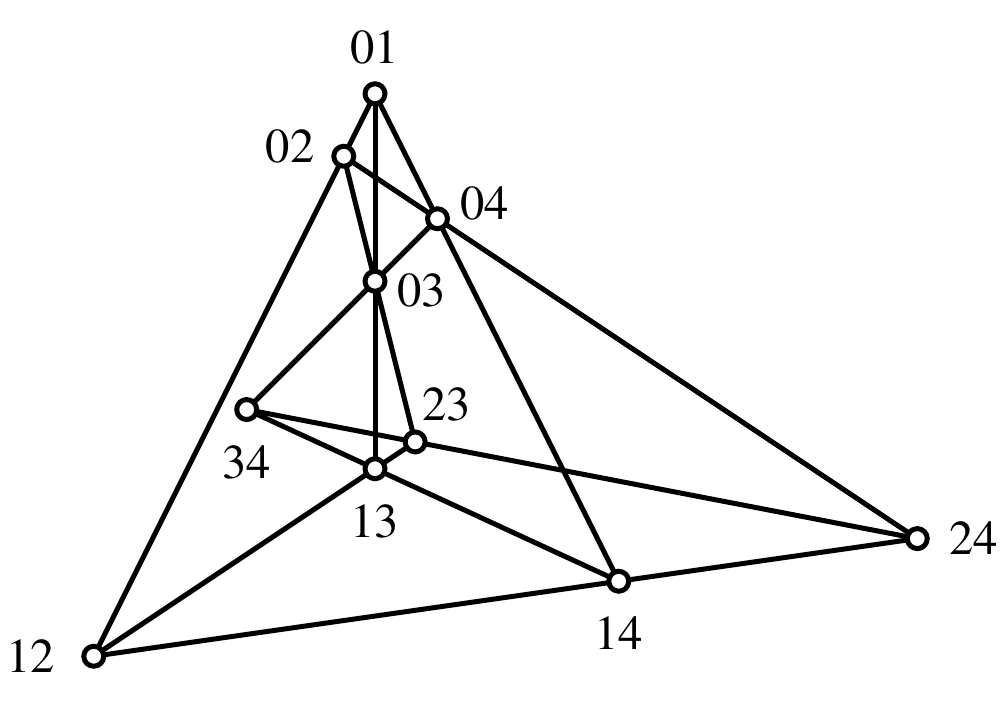}}
\caption{The image of a cell $P(2,4)$ -- Desargues configuration $10_3$. Ten lines correspond to the the white
triangles $\{jk,j\ell,k\ell\}$. One can clearly recognize the images of the five black
tetrahedra $\{ij,ik,i\ell,im\}=T_i\{j,k,\ell,m\}$, and five complete quadrilaterals,
which are images of the octahedral faces $\langle i\rangle=[jk\ell m]$, each one has
six vertices missing one of the indices ($i$ in this case).}\label{fig:Desargues}
\end{figure}

\begin{remark}
The image of the 4-ambo-simplex $P(3,4)$ turns out to be less
interesting. It is a configuration $10_25_4$ in the projective
plane, i.e., five lines in general position with their ten
pairwise intersection points, see Fig.~\ref{fig:P34}. This figure does not support any
non-trivial incidence theorem.
\end{remark}
\begin{figure}[htb]
\centerline{\includegraphics[width=8cm]{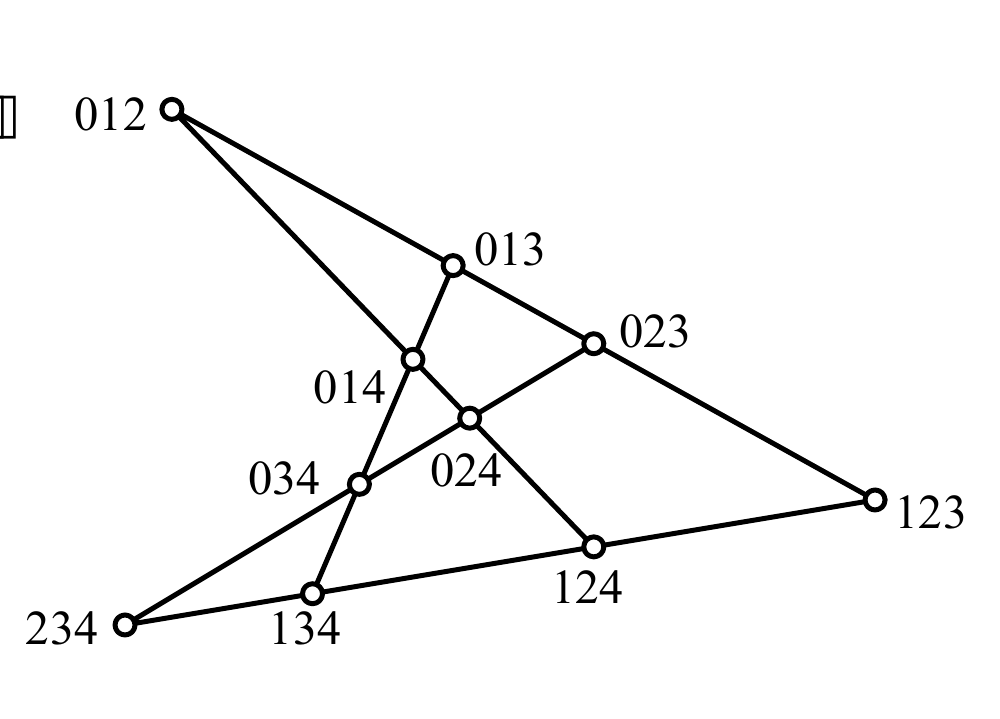}}
\caption{The image of a cell $P(3,4)$ -- a configuration $10_25_4$. Five lines correspond to the the white tetrahedra $\{ijk,ij\ell,ik\ell,jk\ell\}$. One can also recognize five complete quadrilaterals, which are images of the octahedral faces $T_i\langle i\rangle$, each one has six vertices sharing one of the indices ($i$ in this case).}\label{fig:P34}
\end{figure}

\subsection{Consistent triples of octahedron type equations}
\label{s:3}

A local definition of the 4D consistency of octahedron type equations should deal with the 4-cells of $Q(A_4)$ possessing octahedral faces, i.e., with 4-ambo-simplices. We consider in more detail an ambo-simplex whose tetrahedral faces are black. Its five octahedral faces
\[
\langle 0\rangle,\; \langle 1\rangle,\; \langle 2\rangle,\; \langle 3\rangle,\; \langle 4\rangle,
\]
are characterized by the property that any two of them share a triangular face. Two 4-ambo-simplices are said to be {\em adjacent} if they share a common octahedral face. For instance, for any fixed $i=0,1,\ldots,4$, the octahedra
\[
T_0T_i^{-1}\langle 0\rangle,\; T_1T_i^{-1}\langle 1\rangle,\; T_2T_i^{-1}\langle 2\rangle,\; T_3T_i^{-1}\langle 3\rangle,\; T_4T_i^{-1}\langle 4\rangle
\]
serve as the faces of a 4-ambo-simplex adjacent to the previous one, with the common octahedron $\langle i\rangle$.

In the rest of this section, we will work with the representation of the lattice $Q(A_4)$ as
$\mathbb Z^4$, which corresponds to the representation of the octahedron type equations in the form (\ref{eq:oct gen cubic}). This representation appears through a projection from $\mathbb Z^{5}$ to $\mathbb Z^4$ along the 0-th coordinate axis (that is, by forgetting the index $n_0$). It is useful to remember that the shift operators $T_i$ in the $\mathbb Z^4$ representation of $Q(A_4)$ stand for the operators $T_iT_0^{-1}$ in the standard representation (\ref{eq: AN lattice}) of $Q(A_4)$. As already mentioned, we will consider {\em triples} of octahedron type equations corresponding to three 3D sublattices of $\mathbb Z^4$, say the sublattices $(124)$, $(134)$, $(234)$:
\begin{equation}\label{eq: F triple}
\begin{aligned}
 F(x_1,x_2,x_4,x_{12},x_{14},x_{24})&=0,\\
 G(x_1,x_3,x_4,x_{13},x_{14},x_{34})&=0,\\
 H(x_2,x_3,x_4,x_{23},x_{24},x_{34})&=0.
\end{aligned}
\end{equation}
An elementary combinatorial structure supporting these equations is a {\em triple of octahedra}
\begin{equation}\label{eq: oct triple}
\langle 3\rangle,\; \langle 2\rangle,\; \langle 1\rangle,
\end{equation}
which are just three out of five octahedral faces of a 4-ambo-simplex (in particular, any two octahedra of a triple share a triangular face). We will say that another triple of octahedra of the same coordinate directions is {\em adjacent} to the triple (\ref{eq: oct triple}) if they are faces of an adjacent 4-ambo-simplex, and the common octahedral face of the 4-ambo-simplices belongs to neither of the triples. Thus, for two adjacent triples there is a further (seventh) octahedron sharing a face with any octahedron of the both triples. We will say that the both triples of octahedra are adjacent along the seventh one. For instance, the triple
\begin{equation}\label{eq: oct triple shifted 1}
T_3\langle 3\rangle,\; T_2\langle 2\rangle,\; T_1\langle 1\rangle,
\end{equation}
is adjacent to triple (\ref{eq: oct triple}) along the octahedron $\langle 0\rangle$, and the triple
\begin{equation}\label{eq: oct triple shifted 2}
T_3T_4^{-1}\langle 3\rangle,\; T_2T_4^{-1}\langle 2\rangle,\; T_1T_4^{-1}\langle 1\rangle,
\end{equation}
is adjacent to triple (\ref{eq: oct triple}) along the octahedron $\langle 4\rangle$. We are now in a position to formulate a local definition of 4D consistency of the octahedron type equations.

\begin{definition}\label{def: 4D local}
A triple of the octahedron type equations (\ref{eq: F triple}) is called 4D consistent if they can be imposed (admit generic solutions) on any two adjacent triples of octahedra.
\end{definition}

Thus, Definition \ref{def: 4D local} requires to consider equations (\ref{eq: F triple}) on the triple of octahedra (\ref{eq: oct triple}) and on its two adjacent triples, (\ref{eq: oct triple shifted 1}) and (\ref{eq: oct triple shifted 2}).

To verify the first requirement of Definition \ref{def: 4D local}, one can consider the same set of 9 initial data as used in Section \ref{subsect: dKP consist} for the dKP equation, namely
\begin{equation}\label{ini loc 1}
x_1,\; x_2,\; x_3,\; x_4,\; x_{14},\; x_{24},\; x_{34},\; x_{134},\; x_{234}.
\end{equation}
 Equations for the three octahedra (\ref{eq: oct triple}) determine the fields $x_{12}$, $x_{13}$, $x_{23}$:
\begin{equation}\label{xij}
\begin{aligned}
 x_{12}&=f(x_1,x_2,x_4,x_{14},x_{24}),\\
 x_{13}&=g(x_1,x_3,x_4,x_{14},x_{34}),\\
 x_{23}&=h(x_2,x_3,x_4,x_{24},x_{34}),
\end{aligned}
\end{equation}
then equation for the octahedron $T_3\langle 3\rangle$ determines the field $x_{123}$:
\begin{equation}\label{x123 in x}
 x_{123}=\hat{f}(x_{13},x_{23},x_{34},x_{134},x_{234}),
\end{equation}
and finally two equations for the octahedra $T_2\langle 2\rangle$ and $T_1\langle 1\rangle$ deliver two a priori different values for $x_{124}$:
\begin{equation}\label{x124 in x}
\begin{aligned}
   x_{124} &=\hat{g}(x_{12},x_{23},x_{24},x_{123},x_{234})\\
           &=\hat{h}(x_{12},x_{13},x_{14},x_{123},x_{134}).
\end{aligned}
\end{equation}
The first requirement in Definition \ref{def: 4D local} is that these two values of $x_{124}$ identically coincide as functions of the 9 initial data (\ref{ini loc 1}).

One can proceed similarly to verify the second requirement of Definition \ref{def: 4D local}: consider the set of 9 initial data
\begin{equation}\label{ini loc 2}
x_1,\; x_2,\; x_3,\; x_4,\; x_{14},\; x_{24},\; x_{34},\;
x_{13,-4},\; x_{23,-4}.
\end{equation}
One starts with equations (\ref{xij}) for the three octahedra (\ref{eq: oct triple}),  solves equation for the octahedron $T_3T_4^{-1}\langle 3\rangle$ to determine the field $x_{123,-4}$:
\begin{equation}\label{x123-4 in x}
 x_{123,-4}=\tilde{f}(x_3,x_{13},x_{23},x_{13,-4},x_{23,-4}),
\end{equation}
and finally has two equations for the octahedra $T_2T_4^{-1}\langle 2\rangle$ and $T_1T_4^{-1}\langle 1\rangle$ which give two a priori different values for $x_{12,-4}$:
\begin{equation}\label{x12-4 in x}
\begin{aligned}
   x_{12,-4} &=\tilde{g}(x_2,x_{12},x_{23},x_{23,-4},x_{123,-4})\\
           &=\tilde{h}(x_1,x_{12},x_{13},x_{13,-4},x_{123,-4}).
\end{aligned}
\end{equation}
The second requirement in Definition \ref{def: 4D local} is that the two values of $x_{12,-4}$ coincide as functions of the 9 initial data (\ref{ini loc 2}).

For the formulation of the following statements we use the following convention: the derivatives are denoted by the lower indices:
\[
 f_1=\frac{\partial f}{\partial x_1},\quad\dots,\quad
 h_{34}=\frac{\partial h}{\partial x_{34}}
\]
(since we will not need higher order derivatives, the multiple indices
should cause no mis\-understanding).

\begin{statement}\label{st:fgh loc}
If the triple of the octahedron type equations (\ref{xij}) is 4D consistent, then the following holds:
\begin{equation}\label{eq: ranks}
{\rm rk} \begin{pmatrix} f_1 & f_2 & 0 & f_4 \\
 g_1 & 0 & g_3 & g_4 \\ 0 & h_2 & h_3 & h_4 \end{pmatrix}\le 2,\quad
 {\rm rk} \begin{pmatrix} f_4 & f_{14} & f_{24} & 0  \\
 g_4 & g_{14} & 0 & g_{34} \\ h_4 & 0 & h_{24} & h_{34} \end{pmatrix}\le 2.
\end{equation}
\end{statement}
\begin{proof} The first claim follows from the consistency of equations for the triples (\ref{eq: oct triple}), (\ref{eq: oct triple shifted 1}), while the second follows from the consistency of equations for the triples (\ref{eq: oct triple}), (\ref{eq: oct triple shifted 2}). Each of the matrices in (\ref{eq: ranks}) contains the derivatives with respect to the fields {\em not} belonging to the octahedron along which the corresponding triple is adjacent. Both statements are verified similarly, therefore we only prove the first one. We differentiate condition $\hat{g}=\hat{h}$ in (\ref{x124 in x}) with respect to the initial data:
\[
\begin{array}{rl}
 \;\;\partial_{x_1}:& \hat{g}_{12}f_1+\hat{g}_{123}\hat{f}_{13}g_1=
 \hat{h}_{12}f_1+\hat{h}_{13}g_1+\hat{h}_{123}\hat{f}_{13}g_1, \\
 \partial_{x_2}:&
 \hat{g}_{12}f_2+\hat{g}_{23}h_2+\hat{g}_{123}\hat{f}_{23}h_2=
 \hat{h}_{12}f_2+\hat{h}_{123}\hat{f}_{23}h_2, \\
 \partial_{x_3}:&
 \hat{g}_{23}h_3+\hat{g}_{123}(\hat{f}_{13}g_3+\hat{f}_{23}h_3)=
 \hat{h}_{13}g_3+\hat{h}_{123}(\hat{f}_{13}g_3+\hat{f}_{23}h_3), \\
 \partial_{x_4}:& \hat{g}_{12}f_4+\hat{g}_{23}h_4+\hat{g}_{123}(\hat{f}_{13}g_4+\hat{f}_{23}h_4)=
 \hat{h}_{12}f_4+\hat{h}_{13}g_4+\hat{h}_{123}(\hat{f}_{13}g_4+\hat{f}_{23}h_4).
\end{array}
\]
These equations say that the row vector
\[
 \begin{pmatrix} \hat{g}_{12}-\hat{h}_{12}, & -\hat{h}_{13}+(\hat{g}_{123}-\hat{h}_{123})\hat{f}_{13}, &
 \hat{g}_{23}+(\hat{g}_{123}-\hat{h}_{123})\hat{f}_{23}\end{pmatrix}
 \]
 belongs to the left kernel of the matrix
 \[
 \begin{pmatrix} f_1 & f_2 & 0 & f_4 \\
 g_1 & 0 & g_3 & g_4 \\ 0 & h_2 & h_3 & h_4 \end{pmatrix},
\]
so that the latter matrix has rank $\le 2$.
\end{proof}

\begin{statement}\label{st: 3-->5}
If the triple of the octahedron type equations (\ref{xij}) is consistent, then some octahedron type equations are automatically fulfilled on the sublattices $\langle4\rangle=(123)$ and $\langle 0\rangle=(1234)$:
\begin{equation}\label{k}
 K(x_1,x_2,x_3,x_{12},x_{13},x_{23})=0
\end{equation}
and
\begin{equation}\label{l}
 L(x_{12},x_{13},x_{23},x_{14},x_{24},x_{34})=0.
\end{equation}
\end{statement}
\begin{proof}
The fact that
\[
 \mathop{\rm rk}\begin{pmatrix}
   f_1 & f_2 &  0 & f_4 \\
   g_1 &  0  & g_3 & g_4\\
   0  & h_2 & h_3 & h_4
 \end{pmatrix} \le 2
\]
can be reformulated as follows: if we consider (\ref{xij}) as a system for the unknowns $x_1,x_2,x_3,x_4$ and solve the first two equations of this system for $x_1,x_4$, then the substitution of the result into the last equation will cancel out $x_2,x_3$ identically. This is precisely statement (\ref{l}). Relation (\ref{k}) follows similarly. Thus, consistency of a triple of equations on two adjacent triples of octahedra yields that the vertices of the connecting octahedron also satisfy certain equation.
\end{proof}

Thus, the 4D consistency of a triple of the octahedron type equations yields that actually all coordinate directions are on the same footing: each of the five root lattices $Q(A_3)$ contained in $Q(A_4)$ carries its own equation of the octahedron type.

\begin{remark} It is not clear how to extend the initial data (\ref{ini loc 1}) or (\ref{ini loc 2}) in order to get a set of independent initial data for the whole 4D lattice $\mathbb Z^4$. A possible formulation applicable to the whole of $\mathbb Z^4$ is the following: the initial data consist of the values of $x$ on the coordinate 2D planes $(i4)$ for all $i=1,2,3$ (they intersect along the coordinate axis $4$). Within two neighboring 4D cubes, one has the following set of 11 independent initial data:
\begin{equation}\label{ini less loc}
 x_1,\; x_2,\; x_3,\; x_4,\; x_{14},\; x_{24},\; x_{34},\;
x_{44},\; x_{144},\; x_{244},\; x_{344}
\end{equation}
(the last four data coming from the second 4D cube). One can determine $x_{12}$, $x_{13}$, $x_{23}$ from equations (\ref{xij}), and then $x_{124}$, $x_{134}$, $x_{234}$ from equations \eqref{xij} shifted in the $4$-th coordinate direction:
$x_{124}=f(x_{14},x_{24},x_{44},x_{144},x_{244})=T_4(f)$, etc. The 4D consistency conditions appear from the comparison of the three values of $x_{123}$ from \eqref{x123 in x} that must coincide identically as functions of the initial data:
\begin{equation}\label{x123}
 x_{123}=f(g,h,x_{34},T_4(g),T_4(h))=g(f,h,x_{24},T_4(f),T_4(h))=
 h(f,g,x_{14},T_4(f),T_4(g)).
\end{equation}
We will show that this notion of 4D consistency quickly leads to the same necessary conditions formulated in Proposition \ref{st:fgh loc}. Denote the exterior functions $f,g,h$ in (\ref{x123}), as well as their derivatives, by the bar, so that, for instance,
\[
 \bar f=T_3(f)=f(g,h,x_{34},T_4(g),T_4(h)),\quad \bar f_1=T_3(f_1).
\]
The following equations are obtained by differentiating (\ref{x123}):
\begin{align}\label{eq1}
&\left\{\begin{array}{rl}
 \partial_{x_1}:& \bar f_1g_1=\bar g_1f_1=\bar h_2f_1+\bar h_3g_1, \\
 \partial_{x_2}:& \bar f_2h_2=\bar g_1f_2+\bar g_3h_2=\bar h_2f_2, \\
 \partial_{x_3}:& \bar f_1g_3+\bar f_2h_3=\bar g_3h_3=\bar h_3g_3, \\
 \partial_{x_4}:& \bar f_1g_4+\bar f_2h_4=\bar g_1f_4+\bar g_3h_4=\bar h_2f_4+\bar h_3g_4,
\end{array}\right.\\
\label{eq2}
&\left\{\begin{array}{rl}
 \partial_{x_{44}}:& \bar f_{14}T_4(g_4)+\bar f_{24}T_4(h_4)
    =\bar g_{14}T_4(f_4)+\bar g_{34}T_4(h_4)=\bar h_{24}T_4(f_4)+\bar h_{34}T_4(g_4), \\
 \partial_{x_{144}}:& \bar f_{14}T_4(g_{14})=\bar g_{14}T_4(f_{14})=
 \bar h_{24}T_4(f_{14})+\bar h_{34}T_4(g_{14}), \\
 \partial_{x_{244}}:& \bar f_{24}T_4(h_{24})=\bar g_{14}T_4(f_{24})+
 \bar g_{34}T_4(h_{24})=\bar h_{24}T_4(f_{24}), \\
 \partial_{x_{344}}:& \bar f_{14}T_4(g_{34})+\bar f_{24}T_4(h_{34})=
 \bar g_{34}T_4(h_{34})=\bar h_{34}T_4(g_{34}).
\end{array}\right.
\end{align}
Now equations (\ref{eq1}) say that the row vectors
\[
\begin{pmatrix} -\bar g_1, & \bar f_1, & \bar f_2-\bar g_3 \end{pmatrix} \;\;{\rm and}\;\;
\begin{pmatrix} \bar g_1-\bar h_2, & -\bar h_3, & \bar g_3 \end{pmatrix}
\]
belong to the left kernel of the matrix
\[
\begin{pmatrix}
   f_1 & f_2 &  0 & f_4 \\
   g_1 &  0  & g_3 & g_4\\
   0  & h_2 & h_3 & h_4
 \end{pmatrix}.
\]
Analogously, equations (\ref{eq2}) say that the row vectors
\[
\begin{pmatrix} -\bar g_{14}, & \bar f_{14}, & \bar f_{24}-\bar g_{34} \end{pmatrix}\;\;{\rm and}\;\;
\begin{pmatrix} \bar g_{14}-\bar h_{24}, & -\bar h_{34}, & \bar g_{34} \end{pmatrix}
\]
belong to the left kernel of the ($T_4$-shifted) matrix
\[
\begin{pmatrix}
    f_4 & f_{14} & f_{24} &  0     \\
    g_4 & g_{14} &  0     & g_{34} \\
    h_4 & 0      & h_{24} & h_{34}
 \end{pmatrix}.
\]
Thus, both matrices have rank $\le 2$.
\end{remark}


\section{Tripodal forms of the octahedron type equations}\label{s:var3}

We now derive further analytic consequences of the necessary conditions of 4D consistency formulated in Proposition \ref{st:fgh loc}, which can be formulated as
\begin{equation}\label{fgh1}
 f_1g_3h_2+f_2g_1h_3=0,\qquad
 f_2g_3h_4=f_4g_3h_2+f_2g_4h_3,
\end{equation}
and
\begin{equation}\label{fgh2}
 f_{14}g_{34}h_{24}+f_{24}g_{14}h_{34}=0,\qquad
 f_{24}g_{34}h_4=f_4g_{34}h_{24}+f_{24}g_4h_{34},
\end{equation}
respectively.

\begin{statement}\label{st:3leg}
If the triple of octahedron type equations (\ref{xij}) is 4D consistent, then it can be cast into the form
\begin{equation}\label{3leg:1}
\begin{aligned}
 a(x_1,x_4,x_{14})-b(x_2,x_4,x_{24})&=p(x_{12},x_{14},x_{24}),\\
 c(x_3,x_4,x_{34})-a(x_1,x_4,x_{14})&=q(x_{13},x_{14},x_{34}),\\
 b(x_2,x_4,x_{24})-c(x_3,x_4,x_{34})&=r(x_{23},x_{24},x_{34}),
\end{aligned}
\end{equation}
and simultaneously into the form
\begin{equation}\label{3leg:2}
\begin{aligned}
 A(x_1,x_4,x_{14})-B(x_2,x_4,x_{24})&=P(x_1,x_2,x_{12}),\\
 C(x_3,x_4,x_{34})-A(x_1,x_4,x_{14})&=Q(x_1,x_3,x_{13}),\\
 B(x_2,x_4,x_{24})-C(x_3,x_4,x_{34})&=R(x_2,x_3,x_{23}).
\end{aligned}
\end{equation}
\end{statement}
\begin{proof}
We prove that the general solution of equations (\ref{fgh1}) is of the form
\begin{align*}
 f&=\phi(a(x_1,x_4,x_{14})-b(x_2,x_4,x_{24}),x_{14},x_{24}),\\
 g&=\psi(c(x_3,x_4,x_{34})-a(x_1,x_4,x_{14}),x_{14},x_{34}),\\
 h&=\chi(b(x_2,x_4,x_{24})-c(x_3,x_4,x_{34}),x_{24},x_{34}),
\end{align*}
which is obviously equivalent to (\ref{3leg:1}). First equation in (\ref{fgh1}) implies
\[
 \frac{f_1}{f_2}(x_1,x_2,x_4,x_{14},x_{24})=
 -\frac{g_1/g_3(x_1,x_3,x_4,x_{14},x_{34})}
 {h_2/h_3(x_2,x_3,x_4,x_{24},x_{34})}
   =-\frac{\a(x_1,x_4,x_{14})}{\b(x_2,x_4,x_{24})}
\]
(it is sufficient to choose some fixed values of $x_3$ and $x_{34}$ in the middle expression). Now,
\[
 \frac{\a g_3}{g_1}=\frac{\b h_3}{h_2}=-\g(x_3,x_4,x_{34}),
\]
where $\g$ denotes the common value of both ratios. Clearly, $\a,\b,\g$ are defined up to a common factor, possibly depending on $x_4$. We arrive at
\[
 \b f_1+\a f_2=0,\quad \a g_3+\g g_1=0,\quad \g h_2+\b h_3=0.
\]
Choosing some $a=a(x_1,x_4,x_{14})$, $b=b(x_2,x_4,x_{24})$, $c=c(x_3,x_4,x_{34})$ so that  $\a=a_1$, $\b=b_2$, $\g=c_3$, we come to the conclusion that the functions $f,g,h$ can be represented as
\begin{align*}
 f&=\phi(a(x_1,x_4,x_{14})-b(x_2,x_4,x_{24}),x_4,x_{14},x_{24}),\\
 g&=\psi(c(x_3,x_4,x_{34})-a(x_1,x_4,x_{14}),x_4,x_{14},x_{34}),\\
 h&=\chi(b(x_2,x_4,x_{24})-c(x_3,x_4,x_{34}),x_4,x_{24},x_{34}).
\end{align*}
Substituting this into the second equation in (\ref{fgh1}), one finds:
\begin{equation}\label{444}
 \frac{\phi_4}{\phi'}+\frac{\psi_4}{\psi'}+\frac{\chi_4}{\chi'}=0,
\end{equation}
where prime denotes the derivatives with respect to the first arguments of the
functions. Differentiating this equation with respect to $x_1$ and $x_2$ yields:
\[
 \Bigl(\frac{\phi_4}{\phi'}\Bigr)'=
 \Bigl(\frac{\psi_4}{\psi'}\Bigr)'=
 \Bigl(\frac{\chi_4}{\chi'}\Bigr)'=\d(x_4),
\]
where $\d$ is the common value of all expressions. Now, integration yields:
\begin{gather*}
 \frac{\phi_4}{\phi'}=\d(a-b)+\la(x_4,x_{14})-\mu(x_4,x_{24}),\quad
 \frac{\psi_4}{\phi'}=\d(c-a)+\nu(x_4,x_{34})-\la(x_4,x_{14}),\\
 \frac{\chi_4}{\chi'}=\d(b-c)+\mu(x_4,x_{24})-\nu(x_4,x_{34})
\end{gather*}
(the form of the integration constants follows from (\ref{444})). At this point we
use the remaining freedom in the definition of the functions $a,b,c$, which can be changed
to
\[
 \tilde a=k(x_4)a+\ell(x_4,x_{14}),\quad
 \tilde b=k(x_4)b+m(x_4,x_{24}),\quad
 \tilde c=k(x_4)c+n(x_4,x_{34}),
\]
with arbitrary $k,\ell,m,n$. Denoting
\[
\begin{aligned}
 \phi(a-b,x_4,x_{14},x_{24})
  &=\tilde\phi(\tilde a-\tilde b,x_4,x_{14},x_{24}),\\
 \psi(c-a,x_4,x_{14},x_{34})
  &=\tilde\psi(\tilde c-\tilde a,x_4,x_{14},x_{34}),\\
 \chi(b-c,x_4,x_{24},x_{34})
  &=\tilde\chi(\tilde b-\tilde c,x_4,x_{24},x_{34}),
\end{aligned}
\]
it is easy to see that
\[
 \phi'=k\tilde\phi',\quad
 \phi_4=\tilde\phi_4+(k'(a-b)+\ell_4-m_4)\tilde\phi',
\]
so that
\[
 \frac{\phi_4}{\phi'}=\frac{\tilde\phi_4}{k\tilde\phi'}
  +\frac{k'}{k}(a-b)+\frac{\ell_4}{k}-\frac{m_4}{k},
\]
and analogously for $\psi,\chi$. The choice $k'/k=\d$, $\ell_4=k\la$, $m_4=k\mu$,
$n_4=k\nu_4$ leads to
\[
 \tilde\phi_4=\tilde\psi_4=\tilde\chi_4=0,
\]
and this gives the desired representation. The solution of the second pair of
equations (\ref{fgh2}) is of the same structure and leads to representation
(\ref{3leg:2}).
\end{proof}

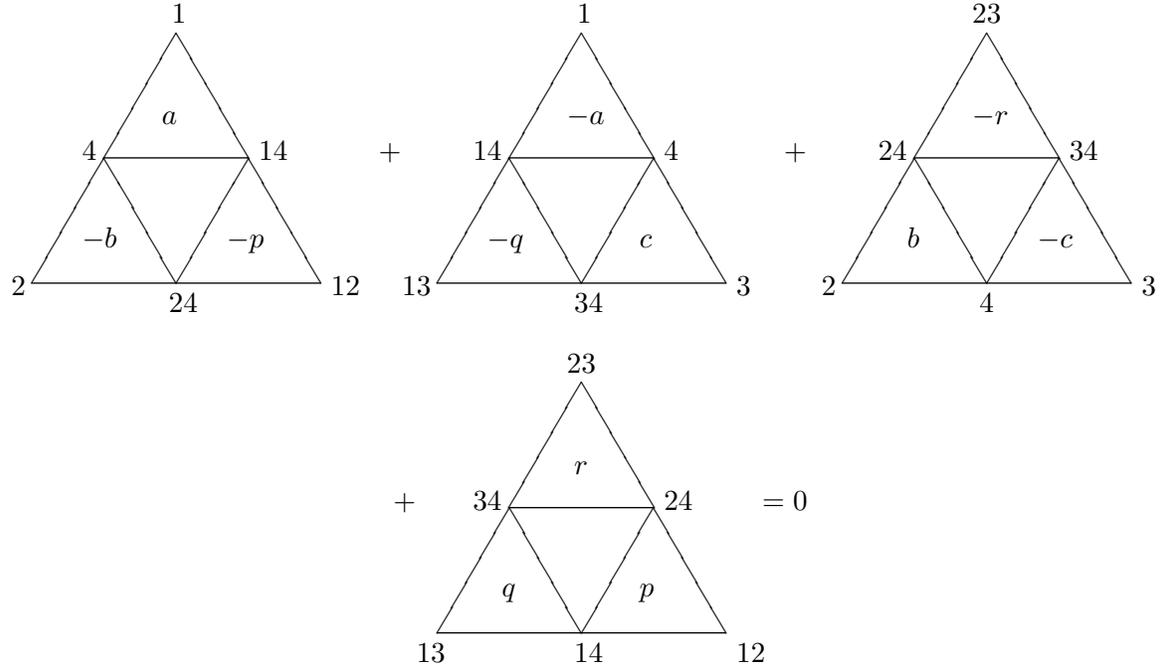
\begin{figure}[htbp]
\setlength{\unitlength}{0.05em}
\centerline{
\begin{picture}(200,200)(0,-30)
 \path(0,0)(200,0)(100,173.2)(0,0)
 \path(100,0)(150,86.6)(50,86.6)(100,0)
 \put(35,85){$4$}\put(207,-8){$12$}
 \put(157,85){$14$}\put(-14,-8){$2$}
 \put(95,-20){$24$}\put(97,180){$1$}
 \put(90,110){$a$}
 \put(35,25){$-b$}
 \put(135,25){$-p$}
 \put(240,85){$+$}
\end{picture}\qquad\qquad
\begin{picture}(200,200)(0,-30)
 \path(0,0)(200,0)(100,173.2)(0,0)
 \path(100,0)(150,86.6)(50,86.6)(100,0)
 \put(25,85){$14$}\put(207,-8){$3$}
 \put(157,85){$4$}\put(-24,-8){$13$}
 \put(95,-20){$34$}\put(97,180){$1$}
 \put(90,110){$-a$}
 \put(35,25){$-q$}
 \put(140,25){$c$}
 \put(240,85){$+$}
\end{picture}\qquad\qquad
\begin{picture}(200,220)(0,-30)
 \path(0,0)(200,0)(100,173.2)(0,0)
 \path(100,0)(150,86.6)(50,86.6)(100,0)
 \put(25,85){$24$}\put(207,-8){$3$}
 \put(157,85){$34$}\put(-14,-8){$2$}
 \put(95,-20){$4$}\put(90,180){$23$}
 \put(90,110){$-r$}
 \put(45,25){$b$}
 \put(135,25){$-c$}
\end{picture}}
\centerline{
\begin{picture}(200,240)(0,-30)
 \path(0,0)(200,0)(100,173.2)(0,0)
 \path(100,0)(150,86.6)(50,86.6)(100,0)
 \put(25,85){$34$}\put(207,-20){$12$}
 \put(157,85){$24$}\put(-14,-20){$13$}
 \put(95,-20){$14$}\put(90,180){$23$}
 \put(95,110){$r$}
 \put(45,25){$q$}
 \put(140,25){$p$}
 \put(-30,85){$+$}
 \put(225,85){$=0$}
\end{picture}}
\caption{Tripodal forms (\ref{3leg:1}) of equations (\ref{xij}) on the octahedra $\langle 3\rangle, \langle 2\rangle, \langle 1\rangle$, and the tripodal form (\ref{eq: pqr}) of equation (\ref{l}) on the octahedron $\langle 0\rangle$ sum up to zero}
\label{fig: 3legs 1}
\end{figure}

We call equations (\ref{3leg:1}) and (\ref{3leg:2}) {\em tripodal forms} of equations (\ref{xij}). For instance, the three terms in the first equation in (\ref{3leg:1}) correspond to the three (triangular) legs (1,4,14), (2,4,14), and (12,14,24) of the tripod with the head (4,14,24), see Fig.~\ref{fig: 3legs 1}. Adding the three tripodal forms (\ref{3leg:1}) of equations on the octahedra $\langle 3\rangle, \langle 2\rangle, \langle 1\rangle$ leads to the equation
\begin{equation}\label{eq: pqr}
p(x_{12},x_{14},x_{24})+q(x_{13},x_{14},x_{34})+r(x_{23},x_{24},x_{34})=0,
\end{equation}
which is nothing but the tripodal form of equation (\ref{l}) on the octahedron $\langle 0\rangle$, with the head (14,24,34).

\begin{figure}[htbp]
\setlength{\unitlength}{0.05em}
\centerline{
\begin{picture}(200,200)(0,-30)
 \path(0,0)(200,0)(100,173.2)(0,0)
 \path(100,0)(150,86.6)(50,86.6)(100,0)
 \put(35,85){$4$}\put(200,-20){$12$}
 \put(157,85){$1$}\put(-14,-20){$24$}
 \put(95,-20){$2$}\put(92,180){$14$}
 \put(90,110){$A$}
 \put(35,25){$-B$}
 \put(130,25){$-P$}
 \put(240,85){$+$}
\end{picture}\qquad\qquad
\begin{picture}(200,200)(0,-30)
 \path(0,0)(200,0)(100,173.2)(0,0)
 \path(100,0)(150,86.6)(50,86.6)(100,0)
 \put(25,85){$1$}\put(200,-20){$34$}
 \put(157,85){$4$}\put(-14,-20){$13$}
 \put(95,-20){$3$}\put(92,180){$14$}
 \put(80,110){$-A$}
 \put(35,25){$-Q$}
 \put(140,25){$C$}
 \put(240,85){$+$}
\end{picture}\qquad\qquad
\begin{picture}(200,220)(0,-30)
 \path(0,0)(200,0)(100,173.2)(0,0)
 \path(100,0)(150,86.6)(50,86.6)(100,0)
 \put(25,85){$2$}\put(200,-20){$34$}
 \put(157,85){$3$}\put(-14,-20){$24$}
 \put(95,-20){$4$}\put(92,180){$23$}
 \put(80,110){$-R$}
 \put(45,25){$B$}
 \put(135,25){$-C$}
\end{picture}}
\centerline{
\begin{picture}(200,240)(0,-30)
 \path(0,0)(200,0)(100,173.2)(0,0)
 \path(100,0)(150,86.6)(50,86.6)(100,0)
 \put(25,85){$3$}\put(200,-20){$12$}
 \put(157,85){$2$}\put(-14,-20){$13$}
 \put(95,-20){$1$}\put(92,180){$23$}
 \put(90,110){$R$}
 \put(45,25){$Q$}
 \put(140,25){$P$}
 \put(-30,85){$+$}
 \put(225,85){$=0$}
\end{picture}}
\caption{Tripodal forms (\ref{3leg:2}) of equations (\ref{xij}) on the octahedra $\langle 3\rangle, \langle 2\rangle, \langle 1\rangle$, and the tripodal form  (\ref{eq: PQR}) of equation (\ref{k}) on the octahedron $\langle 4\rangle$ sum up to zero}
\label{fig: 3legs 2}
\end{figure}
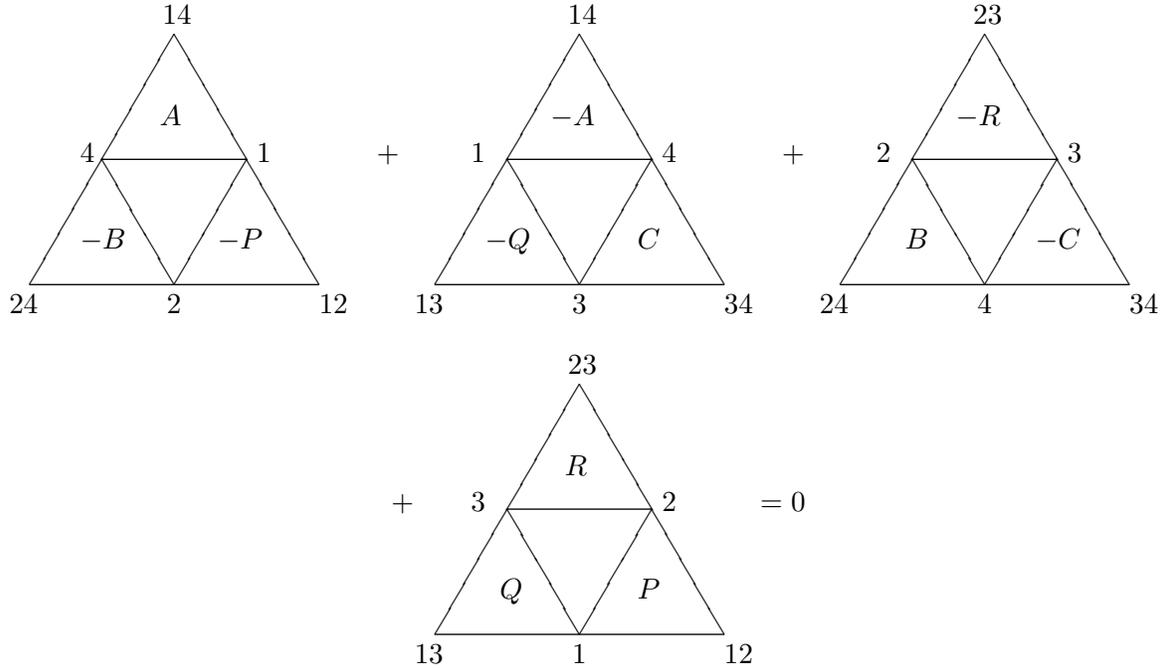

Similarly, adding the three tripodal forms (\ref{3leg:2}) of equations on the octahedra $\langle 3\rangle, \langle 2\rangle, \langle 1\rangle$ leads to the equation
\begin{equation}\label{eq: PQR}
P(x_1,x_2,x_{12})+Q(x_1,x_3,x_{13})+R(x_2,x_3,x_{23})=0,
\end{equation}
which is the tripodal form of equation (\ref{k}) on the octahedron $\langle 4\rangle$, with the head (1,2,3). Thus, Proposition \ref{st:3leg} subsumes (and is actually much stronger than) Proposition \ref{st: 3-->5}.

Moreover, there exist further tripodal representations of equations (\ref{xij}) and (\ref{k}), (\ref{l}), due to the symmetry of all coordinates. Each equation under consideration admits eight tripodal representations, since each face of the octahedron can be chosen as the head of the tripod. In total we have 40 such representations. In order to put them in a unified form, we will use the realization of the $Q(A_4)$ lattice as a hyperplane (\ref{eq: AN lattice}) in $\Integer^5$, and will denote the leg functions in the tripodal form by the letter $a$ with three superscripts which are the indices of the three arguments, omitting the arguments themselves. Our convention will be that this notation is symmetric with respect to the first and the third argument (the base of the leg), while the second argument (the spike of the leg) enters the equation only once. In this notation, formulas of Proposition \ref{st:3leg} can be written so: for the quadruple of octahedra $\langle i\rangle$, $\langle j\rangle$, $\langle k\rangle$, $\langle m\rangle$, the tripodal forms
\begin{gather}
    \EQ{i}:\quad a^{jn,jm,mn}-a^{kn,km,mn}=a^{jn,jk,kn},\nonumber\\[0.3em]
    \EQ{j}:\quad a^{kn,km,mn}-a^{in,im,mn}=a^{kn,ik,in},\label{ijk}\\[0.3em]
    \EQ{k}: \quad  a^{in,im,mn}-a^{jn,jm,mn}=a^{in,ij,jn}\nonumber
\end{gather}
sum up to the tripodal form
\[
\EQ{m}:\quad a^{jn,jk,kn}+a^{kn,ik,in}+a^{in,ij,jn}=0.
\]

The following theorem summarizes the results of the present section.

\begin{theorem}\label{th:3leg}
A quintuple of the octahedron type equations is 4D consistent if and only if, in any quadruple of octahedra such that any two of them share a triangular face, the tripodal forms of any three equations sum up to the tripodal form of the fourth one.
\end{theorem}
\begin{proof}
The necessity follows from Proposition \ref{st:3leg}, the sufficiency is proved by literally the same argument as used in Section \ref{subsect: dKP consist} for the proof of the 4D consistency of the dKP equation.
\end{proof}


\section{Tripodal equations}\label{s:3leg}
\subsection{Definitions and notation}\label{s:octahedron}

We have shown that each equation of a 4D consistent system can be written in 8 ways in a tripodal form. In the present section, we temporarily forget about consistency and analyze just this property of a single octahedron type equation. The enumeration of the vertices of an octahedron and of the corresponding variables will be different from the previously used: it will be convenient to enumerate them just from 1 to 6 as on Fig.~\ref{fig:octahedron}.

Our task here will be to describe all equations
\begin{equation}\label{xxxxxx}
 \Phi(x_1,x_2,x_3,x_4,x_5,x_6)=0,
\end{equation}
which are locally equivalent to equations of the form
\[
 a(x_i,x_{7-k},x_j)+b(x_j,x_{7-i},x_k)+c(x_k,x_{7-j},x_i)=0
\]
for each of the eight triples $(i,j,k)$, corresponding to the faces of an octahedron. Such equations will be called {\em tripodal}.

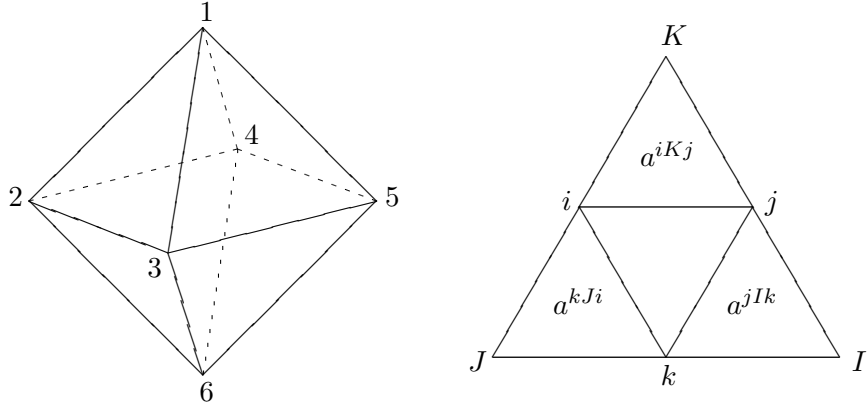
\begin{figure}[t]
\setlength{\unitlength}{0.06em}
\centerline{
\begin{picture}(200,240)(-100,-120)
 \path(-100,0)(0,100)(100,0)(0,-100)(-100,0)(-20,-30)(100,0)
 \path(0,100)(-20,-30)(0,-100)
 \dashline{3}(0,100)(20,30)(0,-100)
 \dashline{3}(-100,0)(20,30)(100,0)
 \put(-2,104){1}\put(-2,-115){6}
 \put(-112,-3){2}\put(105,-3){5}
 \put(-32,-44){3}\put(24,33){4}
\end{picture}\qquad\qquad
\begin{picture}(200,240)(0,-30)
 \path(0,0)(200,0)(100,173.2)(0,0)
 \path(100,0)(150,86.6)(50,86.6)(100,0)
 \put(40,85){$i$}\put(207,-8){$I$}
 \put(157,85){$j$}\put(-14,-8){$J$}
 \put(97,-16){$k$}\put(97,180){$K$}
 \put(85,110){$a^{iKj}$}
 \put(35,25){$a^{kJi}$}
 \put(135,25){$a^{jIk}$}
\end{picture}}
\caption{Left: enumeration of the vertices of an octahedron: (1,2,3) is a face, and opposite vertices carry complementary indices which sum up to 7. Right: the tripodal form with the head $(i,j,k)$.}
\label{fig:octahedron}
\end{figure}

As usual, it is supposed that equation (\ref{xxxxxx}) is irreducible, in particular, neither of the partial derivatives $\Phi_i$ vanishes identically. The answer is determined up to point transformations
\begin{equation}\label{tix}
 \tilde x_i=X_i(x_i),\quad X'_i\ne0.
\end{equation}
Functions $a,b,c$ may be different for different faces. It will be convenient to denote them by superscripts pointing to the arguments, for instance,
$a^{123}=a^{123}(x_1,x_2,x_3)$. Here the second index is distinguished, and it is supposed that the notation is symmetric with respect to the first and the third index: $a^{123}=a^{321}$. This means that $a^{123}(x_1,x_2,x_3)=a^{321}(x_3,x_2,x_1)$ (but of course it is not supposed that
$a^{123}(x_1,x_2,x_3)=a^{321}(x_1,x_2,x_3)$). Opposite vertices of an octahedron will be denoted by the same letter in the lower and the upper cases, e.g., $I=7-i$. Thus, the tripodal form of equation with the head $(ijk)$ takes the form
\[
 a^{iKj}+a^{jIk}+a^{kJi}=0,
\]
see Fig.~\ref{fig:octahedron}.
The set of all eight tripodal forms looks like this:
\begin{equation}\label{3legs}
\begin{array}{cll}
 \text{\em head:} &\qquad & \text{\em tripodal form:} \\[1.5mm]
  (1,2,3) && a^{142}+a^{263}+a^{351}=0, \\
  (1,2,4) && a^{132}+a^{264}+a^{451}=0, \\
  (1,3,5) && a^{123}+a^{365}+a^{541}=0, \\
  (1,4,5) && a^{124}+a^{465}+a^{531}=0, \\
  (2,3,6) && a^{213}+a^{356}+a^{642}=0, \\
  (2,4,6) && a^{214}+a^{456}+a^{632}=0, \\
  (3,5,6) && a^{315}+a^{546}+a^{623}=0, \\
  (4,5,6) && a^{415}+a^{536}+a^{624}=0.
\end{array}
\end{equation}
Clearly, this structure is rather restrictive. It turns out to be possible to find a complete list of tripodal equations. It is given and discussed in the following subsection, while the rest of the section is devoted to the proof that the list is exhaustive, indeed.

\subsection{Classification of tripodal equations}\label{ss:3leg_types}

\begin{theorem}\label{th:3leg_types}
All tripodal equations (\ref{xxxxxx}), up to the point transformations (\ref{tix}) and transpositions $i\leftrightarrow I$, $(i,j)\leftrightarrow(J,I)$, are given by the following list:
\begin{align}
\label{Y1}\tag{\mbox{T$_1$}}
 & x_1x_6+x_2x_5+x_3x_4=0,\\
\label{Y3}\tag{\mbox{T$_2$}}
 & (x_1-x_4)(x_2-x_6)(x_3-x_5)+(x_4-x_2)(x_6-x_3)(x_5-x_1)=0,\\
 \label{Y2}\tag{\mbox{T$_3$}}
 & (x_1-x_2)x_4+(x_2-x_3)x_6+(x_3-x_1)x_5=0,\\
\label{Y4}\tag{\mbox{T$_4$}}
 & x_1x_6=(x_2+x_3)^{-\g}(x_4+x_5),\\
\label{Y5}\tag{\mbox{T$_5$}}
 & x_1x_6=x_2+x_3+x_4+x_5,\\
\label{Y6}\tag{\mbox{T$_6$}}
 & x_1x_2x_3x_4=x_5+x_6,\\
\label{Y7}\tag{\mbox{T$_7$}}
 & x_1+x_2+x_3+x_4+x_5+x_6=0.
\end{align}
\end{theorem}

It is easy to see that the first three equations are equivalent to
(\ref{h1}), (\ref{h3}), (\ref{h2}), upon re-naming
$(x_{12},x_{13},x_{23})\to(x_4,x_5,x_6)$. The proof of the theorem will be given in the following subsections. Here we will just check that all equations of the list are tripodal, indeed, that is, admit all eight representations (\ref{3legs}).

The results of this check are summarized in Table~\ref{fig:sklad}, which lists the functions $a(x,y,z)$ acting as the legs $a^{ijk}=a(x_i,x_j,x_k)$ for equations of the list. Recall that the variables $x,z$ are on the equal footing, while $y$ plays a distinguished role. In particular, the functions $a$ always depend on $y$ but is sometimes independent on $x$ or $z$. Functions $a$ are defined up to point transformations (\ref{tix}). Besides, they can be multiplied by an arbitrary constant, one can flip $x,z$ and add arbitrary combinations of the type $\mu(x)+\nu(z)$ (as long as one such function is considered, and not all three legs of the tripodal form). At this stage, we are not concerned with adjusting all these transformations for consistent equations, therefore the table contains arbitrary (possibly simple) representatives for the leg functions. This arbitrariness notwithstanding, we will see in Section \ref{s:class} that this table is very useful for putting separate tripodal equations into consistent quintuples. The third column of the table contains the class according to the subdivision at the end of Section \ref{ss:var2}.

\begin{figure}[t]
\[
 \begin{array}{|c|c|c|}
 \hline
  \text{Eq.}& \text{legs} & \text{case}\\[0.2em]
 \hline &&\\[-0.7em]
 (\ref{Y1}) & xyz & \text{I} \\[1em]
 (\ref{Y3}) & \log\Bigl(\dfrac{x+y}{y+z}\Bigr) & \text{III} \\[1em]
 (\ref{Y2}) & (x+z)y,\quad \log\Bigl(\dfrac{x+y}{y+z}\Bigr),\quad
              \log(x+y) & \text{III} \\[1em]
 (\ref{Y4}) & xy,\quad y(x+z)^\g,\quad y(x+z)^{1/\g},\quad \log(x+y),\quad y
    & \text{I ($\g\ne1$),  III ($\g=1$)}\\[0.7em]
 (\ref{Y5}) & y,\quad (x+y)z & \text{II} \\[0.7em]
 (\ref{Y6}) & xyz,\quad xy,\quad y,\quad y+\log(x+z),\quad \log(x+y) & \text{I} \\[0.7em]
 (\ref{Y7}) & y & \text{III} \\[0.5em]
 \hline
 \end{array}
\]
\caption{Legs store}\label{fig:sklad}
\end{figure}

{\em Equation} (\ref{Y1}). The tripodal form with the head $(1,2,3)$ is
\[
 (1,2,3):\qquad \frac{x_4}{x_1x_2}+\frac{x_6}{x_2x_3}+\frac{x_5}{x_3x_1}=0,
\]
other ones are obtained by reflections $x_i\leftrightarrow x_I$ and
$(x_i,x_j)\leftrightarrow(x_J,x_I)$, which generate the symmetry group of an octahedron and leave the equation invariant.

{\em Equation} (\ref{Y3}). The multiplicative tripodal form $(1,2,3)$ can be seen directly from the equation. All other ones follow again by reflections $x_i\leftrightarrow x_I$ and
$(x_i,x_j)\leftrightarrow(x_J,x_I)$, although the invariance of the equation under these reflections is less obvious than in the case of equation \ref{Y1}.

{\em Equation} (\ref{Y2}) is already written in the tripodal form $(123)$. Arranging the terms in a different order leads to the tripodal form for the opposite face:
\[
(4,5,6): \qquad x_1(x_4-x_5)+x_3(x_5-x_6)+x_2(x_6-x_4)=0.
\]
Thus, the equation is invariant under the involution
\[
 P:\quad (x_1,x_2,x_3)\leftrightarrow(x_6,x_5,x_4).
\]
Further, one finds the {\em multiplicative} tripodal form
\[
 (1,2,4):\qquad \Bigl(\frac{x_1-x_3}{x_3-x_2}\Bigr)\frac{(x_4-x_5)}{(x_6-x_4)}=1.
\]
All other tripodal forms are obtained with the help of the involution $P$ and the cyclic permutation
\[
 Z:\quad x_1\to x_2\to x_3\to x_1,\quad x_6\to x_6\to x_4\to x_6.
\]

{\em Equation} (\ref{Y4}) is invariant under the transformations
\[
 x_1\leftrightarrow x_6;\qquad
 (x_2,x_5)\leftrightarrow(x_3,x_4);\qquad
 (x_1,x_2,x_3,\g)\leftrightarrow(x^{-1/\g}_1,x_5,x_4,-1/\g),
\]
which allows one to obtain all the tripodal forms from the following two:
\begin{alignat*}{2}
 &(1,2,3):&\qquad & \frac{x_4}{x_1}-x_6(x_2+x_3)^\g+\frac{x_5}{x_1}=0,\\
 &(1,2,4):& & (\log x_1+\g\log(x_2+x_3))+\log x_6-\log(x_4+x_5)=0.
\end{alignat*}

{\em Equation} (\ref{Y5}). In this case the symmetry group is generated by the transformations
\[
 x_i\leftrightarrow x_I;\qquad
 (x_2,x_5)\leftrightarrow(x_3,x_4),
\]
which allows one to obtain all the tripodal forms from one of them, for instance, from
\[
 (1,2,3):\qquad \frac{x_4+x_2}{x_1}-x_6+\frac{x_3+x_5}{x_1}=0.
\]
It is worth mentioning that, although here the first and the third terms are similar to the legs encountered for equations (\ref{Y2}) and (\ref{Y4}) with $\g=1$, this is a different type of legs, since the role of the distinguished variable $y$ is different.

{\em Equation} (\ref{Y6}). Here there are transformations
\[
 (x_1,x_6)\leftrightarrow(x_2,x_5);\qquad x_3\leftrightarrow x_4,
\]
which make it sufficient to give the following three tripodal forms:
\begin{alignat*}{2}
 &(1,2,3):&\qquad & x_1x_4x_2-\frac{x_6}{x_3}-\frac{x_5}{x_3}=0,\\
 &(1,3,5):& & \log(x_1x_2x_3)-\log(x_5+x_6)+\log x_4=0,\\
 &(3,5,6):& & (\log x_3+\log x_1)+(\log x_4-\log(x_5+x_6))+\log x_2=0.
\end{alignat*}

{\em Equation} (\ref{Y7}). This case is trivial.

\subsection{Reduction to functions of two variables}\label{ss:var2}

Let us prove some useful relations which follow from the (local) equivalence of the two tripodal forms
\begin{equation}\label{ikKj}
 (i,j,K):~~a^{ikj}+a^{jIK}+a^{KJi}=0,\qquad (i,j,k):~~a^{iKj}+a^{jIk}+a^{kJi}=0,
\end{equation}
which correspond to two faces sharing a common edge $ij$. We will use the subscripts $i$ to denote the partial derivatives $\partial_{x_i}$.

\begin{statement}\label{st:id}
The following identities hold true:
\begin{gather}
\label{id1}
 \frac{a^{ikj}_i+a^{KJi}_i}{a^{KJi}_J}=\frac{a^{iKj}_i+a^{kJi}_i}{a^{kJi}_J},\qquad
 \frac{a^{ikj}_j+a^{jIK}_j}{a^{jIK}_I}=\frac{a^{iKj}_j+a^{jIk}_j}{a^{jIk}_I},\\
\label{id2}
 \Ifrac{a^{iKj}_K}{j} \Ifrac{a^{KJi}_J}{K}=
 \Ifrac{a^{ikj}_k}{j} \Ifrac{a^{kJi}_J}{k},\qquad
 \Ifrac{a^{iKj}_K}{i} \Ifrac{a^{jIK}_I}{K}=
 \Ifrac{a^{ikj}_k}{i} \Ifrac{a^{jIk}_I}{k}.
\end{gather}
Since each one contains five variables only, they are satisfied identically and not by virtue of equations (\ref{ikKj}).
\end{statement}
\begin{proof}
The numerators do not vanish identically because of the local solvability of equations with respect to each variable. We shall prove only the first identities in each pair, since the second ones follow by the permutation $(i,I)\leftrightarrow(j,J)$. Note also that all identities are invariant under the flip $k\leftrightarrow K$. Solving equations (\ref{ikKj}) for $x_I$, we obtain an equation of the following form:
\[
 f(x_j,x_K,a^{ikj}+a^{KJi})=g(x_j,x_k,a^{iKj}+a^{kJi}),
\]
and by differentiation there follow equations of the form
\[
 f_K+f'a^{KJi}_K=g'a^{iKj}_K,\quad
 f'(a^{ikj}_i+a^{KJi}_i)=g'(a^{iKj}_i+a^{kJi}_i),\quad
 f'a^{KJi}_J=g'a^{kJi}_J,
\]
where the prime denotes the derivatives of the functions $f,g$ with respect to their last arguments. Now (\ref{id1}) follows immediately from the second and the third equations, while from the first and the third we derive
\[
 \frac{f_K}{f'}=\frac{a^{KJi}_J}{a^{kJi}_J}a^{iKj}_K-a^{KJi}_K=\phi
 \quad\Rightarrow\quad \phi_Ja^{ikj}_k=\phi_ka^{KJi}_J,
\]
which yields
\[
 \biggl(\biggl(\frac{a^{KJi}_J}{a^{kJi}_J}\biggr)_{\!\!J}a^{iKj}_K
  -a^{KJi}_{KJ}\biggr)a^{ikj}_k
 =(a^{KJi}_J)^2a^{iKj}_K\Ifrac{a^{kJi}_J}{k}.
\]
Dividing by $a^{iKj}_Ka^{ikj}_k(a^{KJi}_J)^2$ and differentiating with respect to $x_j$, we arrive at (\ref{id2}).
\end{proof}

Equations (\ref{id1}) will play the main role in the subsequent analysis. Equations
(\ref{id2}) are needed only for the proof of the following property, which, in its turn, will only be used in the proof of Proposition \ref{st:T23}. Nevertheless, it is convenient to do the job right now, since all necessary formulas are at hand.

\begin{statement}\label{st:imp2}
If at least one of the functions $a^{jIK}_{IK}$, $a^{KJi}_{KJ}$, $a^{jIk}_{Ik}$,
$a^{kJi}_{kJ}$ does not vanish identically, then
\begin{equation}\label{Aij}
 \frac{a^{ikj}_{ik}a^{ikj}_{kj}}{(a^{ikj}_k)^2}=
 \frac{a^{iKj}_{iK}a^{iKj}_{Kj}}{(a^{iKj}_K)^2}.
\end{equation}
\end{statement}
\begin{proof}
Differentiating (\ref{id1}) with respect to $x^k$ and $x^K$, we obtain equations
\[
 a^{ikj}_{ik}\Ifrac{a^{KJi}_J}{K}=a^{iKj}_{iK}\Ifrac{a^{kJi}_J}{k},\qquad
 a^{ikj}_{kj}\Ifrac{a^{jIK}_I}{K}=a^{iKj}_{Kj}\Ifrac{a^{jIk}_I}{k}.
\]
Along with (\ref{id2}) they build a pair of linear homogeneous systems, one for
$(1/a^{KJi}_J)_K$, $(1/a^{kJi}_J)_k$ and another for $(1/a^{jIK}_I)_K$,
$(1/a^{jIk}_I)_k$. Determinants of these systems coincide. The statement of the proposition says that if at least one of the systems has a non-trivial solution, then their common determinant vanishes.
\end{proof}

The following representation is a direct consequence of identities (\ref{id1}).

\begin{statement}\label{st:aij}
The functions $a^{ikj}$ and $a^{iKj}$ are of the form
\begin{equation}\label{aikj}
 a^{ikj}=a^{ij}b^k+p^{ik}+p^{kj},\qquad a^{iKj}=a^{ij}b^K+p^{iK}+p^{Kj},\qquad
 b^kb^K\ne0.
\end{equation}
\end{statement}
\begin{proof}
Differentiating equations (\ref{id1}), we find:
\begin{equation}\label{aaaa}
 a^{ikj}_{ij}a^{kJi}_J=a^{iKj}_{ij}a^{KJi}_J,\qquad
 a^{ikj}_{ij}a^{jIk}_I=a^{iKj}_{ij}a^{jIK}_I.
\end{equation}
Setting the variables $x_I,x_J,x_K$ to arbitrary constants, we find:
\[
 a^{ikj}_{ij}=\a^{ij}\b^{ik}=\a^{ij}\g^{jk}\quad\Rightarrow\quad
 a^{ikj}_{ij}=\a^{ij}b^k,
\]
and an integration leads to formula (\ref{aikj}). Here the function $a^{ij}$ is defined up to addition of the terms $\mu^i+\nu^j$ and multiplication by a constant, which can be taken into account by a re-definition of $b$ and $p$. It is easy to see from (\ref{aaaa}) that this function can be chosen the same in the both formulas, and that the factors $b$ can be taken non-vanishing without restriction of generality.
\end{proof}

\begin{remark}
Taking into account all the tripodal forms (\ref{3legs}), we come to the conclusion that each function $a^{lmn}$ has a representation of the type (\ref{aikj}), so that the form of the equation is already found up to functions of two variables. However it should be understood that the notation in formulas (\ref{aikj}) refer to some fixed pair of the tripodal forms (\ref{ikKj}). If we would like to apply them to all possible index sets simultaneously then we would be forced to use more complicated enumeration for the functions $b$ and $p$. However, there will be no need to do this, since we will only perform a pairwise comparison of the tripodal forms.
\end{remark}

\begin{statement}\label{st:branch}
If $a^{ij}_{ij}\ne0$ then either the both functions $b^k$ and $b^K$ are different from constants or the both are constants, the tripodal forms (\ref{ikKj}) being either
\begin{equation}\label{ikKj-I}
 a^{ij}b_k+\frac{c^{jI}}{b_K}+\frac{c^{Ji}}{b_K}=0\quad\Leftrightarrow\quad
 a^{ij}b_K+\frac{c^{jI}}{b_k}+\frac{c^{Ji}}{b_k}=0
\end{equation}
or
\begin{equation}\label{ikKj-II}
\begin{gathered}
 (a^{ij}+p^{ik}+p^{kj})+(c^{jI}+p^{Kj})+(c^{Ji}+p^{iK})=0\\
 \Leftrightarrow\qquad
 (a^{ij}+p^{iK}+p^{Kj})+(c^{jI}+p^{kj})+(c^{Ji}+p^{ik})=0,
\end{gathered}
\end{equation}
respectively.
\end{statement}
\begin{proof}
We still did not exhaust the content of identities (\ref{aaaa}). If $a^{ij}_{ij}\ne0$ then they are reduced to
\begin{equation}\label{baba}
 b^ka^{kJi}_J-b^Ka^{KJi}_J=0,\qquad b^ka^{jIk}_I-b^Ka^{jIK}_I=0,
\end{equation}
which give, upon integration,
\begin{alignat}{3}
\label{akJijIk}
 a^{kJi}&=\frac{c^{Ji}}{b^k}+d^{ki},&\qquad a^{jIk}&=\frac{c^{jI}}{b^k}+d^{jk},\\
\label{aKJijIK}
 a^{KJi}&=\frac{c^{Ji}}{b^K}+d^{Ki},&\qquad a^{jIK}&=\frac{c^{jI}}{b^K}+d^{jK}.
\end{alignat}
Thus, the tripodal forms (\ref{ikKj}) can be put as
\begin{align*}
 (a^{ij}b^k+p^{ik}+p^{kj})
 +\Bigl(\frac{c^{jI}}{b^K}+d^{jK}\Bigr)
 +\Bigl(\frac{c^{Ji}}{b^K}+d^{Ki}\Bigr)=0,\\
 (a^{ij}b^K+p^{iK}+p^{Kj})
 +\Bigl(\frac{c^{jI}}{b^k}+d^{jk}\Bigr)
 +\Bigl(\frac{c^{Ji}}{b^k}+d^{ki}\Bigr)=0.
\end{align*}
It is easy to see that they are equivalent if and only if
\[
 b^K(p^{ik}+p^{kj}+d^{jK}+d^{Ki})=b^k(p^{iK}+p^{Kj}+d^{jk}+d^{ki}),
\]
from which it follows $b^K_K(p^{ik}_k+p^{kj}_k)=b^k_k(p^{iK}_K+p^{Kj}_K)$. Now if
$b^k_k=0$, then $p^{ik}_k+p^{kj}_k\ne0$, since otherwise the first tripodal form would not contain $x_k$, but then also $b^K_K=0$, and, setting $b^k=b^K=1$, we find formulas (\ref{ikKj-II}).

If, on the contrary, $b^k_k\ne0$, then
\[
 p^{ik}+p^{kj}=b^k(p^i+p^j)+q^i+q^j,\qquad
 p^{iK}+p^{Kj}=b^K(p^i+p^j)+\tilde q^i+\tilde q^j.
\]
It is not difficult to see that these terms can be assumed to vanish (this can be achieved by a re-definition of other terms). But then
\[
 d^{jK}+d^{Ki}=\frac{d^j+d^i}{b^K},\qquad d^{jk}+d^{ki}=\frac{d^j+d^i}{b^k},
\]
and these terms can be again absorbed by the functions $c^{jI}$, $c^{Ji}$. As a result, the tripodal forms (\ref{aikj}) can be put as (\ref{ikKj-I}).
\end{proof}

The further analysis does not require for complicated computations, but rather for a detailed bookkeeping of different cases. Taking into account the possibilities listed in Proposition \ref{st:branch}, one can take care of this by separating the following classes of equations:\medskip

\begin{tabular}{ll}
I. & $a^{ikj}_{ikj}\ne0$ for at least one triple $i,k,j$;\\[0.2em]
II. & $a^{ikj}_{ikj}=0$ for all $i,k,j$, but $a^{ikj}_{ij}\ne0$ for at least one triple;\\[0.2em]
III. &$a^{ikj}_{ij}=0$ for all $i,k,j$.
\end{tabular}\medskip

\noindent Clearly, this exhausts all logical possibilities. A complete description of these cases will be sufficient for a proof of Theorem \ref{th:3leg_types}.

\subsection{Class I}\label{ss:caseI}

Setting, without restriction of generality, $b^k=x_k$, $b^K=x_K$, we put the tripodal forms  (\ref{ikKj-I}) as
\begin{equation}\label{Igen}
 a^{ij}x_k+\frac{c^{jI}}{x_K}+\frac{c^{Ji}}{x_K}=0,
 \quad\Leftrightarrow\quad
 a^{ij}x_K+\frac{c^{jI}}{x_k}+\frac{c^{Ji}}{x_k}=0.
\end{equation}
Now we have to compare them with the other tripodal forms. Since functions $a^{jIk}$, $a^{kJi}$ have to be of the form (\ref{aikj}), as well, we immediately find that
\[
  c^{jI}=c^jb^I+d^I,\quad c^{Ji}=c^ib^J+d^J.
\]
Assume first that $c^{jI}_J\ne0$ or $c^{Ji}_i\ne0$. For definiteness, let $c^{jI}_J\ne0$. Compare (\ref{Igen}) with the tripodal form $(I,j,k)$:
$a^{IKj}+a^{jik}+a^{kJI}=0$. Apply formulas (\ref{akJijIk}), from which there follows:
\[
 a^{jik}=\frac{c^j}{x_k}b^i+q^{ji}+q^{jk},
\]
while the function $a^{iKj}$ should admit the both representations
\[
 a^{iKj}=a^{ij}x_K=\frac{c^{Kj}}{b^i}+d^{ij}\qquad\Rightarrow\qquad
 a^{iKj}=\frac{x_K}{b^ib^j}.
\]
Since, by assumption, $a^{ij}_{ij}\ne0$, there follows $b^i_i\ne0$, but then there holds also $b^I_I\ne0$, and the equation is reduced to the form
\[
 \frac{c^{Kj}}{b^i}+a^{jk}b^I+\frac{c^{kJ}}{b^i}=0.
\]
Subtracting this from (\ref{Igen}), we come to
\[
 \frac{c^{jI}}{x_k}+\frac{c^{Ji}}{x_k}=a^{jk}b^I+\frac{c^{kJ}}{b^i},
\]
which yields that the equation is reduced by the point transformations (\ref{tix}) to the form (\ref{Y1}).

Now let $c^{jI}_j=c^{Ji}_i=0$, then equation (\ref{Igen}) takes the form
\[
 \log a^{ij}+\log x_K-\log(c_I+c_J)+\log x_k=0.
\]
This case is covered (after an obvious change of enumeration) by the following proposition which refers to somewhat more general equations. This is done in order to avoid unnecessary repetitions for the cases II and III, which also lead to equations of this type. In this proposition it is not assumed that the equation falls into one of the three classes separated above. Actually, the analysis of results performed in Section \ref{ss:3leg_types} shows that equation (\ref{Y6}) belongs to the class I, equation (\ref{Y7}) belongs to the class III, and equation (\ref{Y4}) belongs to the class I if $\la\ne1$ and to the class III if $\la=1$.

\begin{statement}\label{st:spec_1}
Let one of the tripodal forms of equation (\ref{xxxxxx}) be as follows:
\begin{equation}\label{spec_1}
 p^{IK}+p^{Ki}+p^{ik}+q^j+q^J=0.
\end{equation}
Then it is reduced to one of the equations (\ref{Y4}), (\ref{Y6}), or (\ref{Y7}).
\end{statement}
\begin{proof}
Setting, without restriction of generality, $q^j=x_j$, $q^J=x_J$, we put the equation as \begin{equation}\label{IVgen}
 p^{IK}+p^{Ki}+p^{ik}+x_j+x_J=0.
\end{equation}
Let us show that if $p^{Ki}_{Ki}\ne0$, then $p^{IK}_{IK}=0$ and $p^{ik}_{ik}=0$.

In order to prove the first of these equations, compare two tripodal forms (\ref{ikKj}),
\[
 (i,j,K):~~a^{ikj}+a^{jIK}+a^{KJi}=0,\qquad (i,j,k):~~a^{iKj}+a^{jIk}+a^{kJi}=0,
\]
assuming that (\ref{IVgen}) is the first one of them, so that $a^{ikj}=p^{ik}+x_j$,
$a^{jIK}=p^{IK}$ and $a^{KJi}=p^{Ki}+x_J$. From the first identity (\ref{id1}) we find:
\begin{equation}\label{IV-id}
 p^{ik}_i+p^{Ki}_i=\frac{a^{iKj}_i+a^{kJi}_i}{a^{kJi}_J}\quad\Rightarrow\quad
 p^{Ki}_{Ki}=\frac{a^{iKj}_{iK}}{a^{kJi}_J}.
\end{equation}
Suppose that $p^{Ki}_{Ki}\ne0$. Then, differentiating the latter equation with respect to
$x_k$, $x_J$, we find $a^{kJi}_{kJ}=a^{kJi}_{JJ}$, but then the first equations implies also $a^{kJi}_{Ji}=0$. Therefore, $a^{kJi}=\mu x_J+r^{ki}$,  so the tripodal form $(i,j,k)$ becomes
\[
 a^{iKj}+a^{jIk}+\mu x_J+r^{ki}=0,\quad \mu\ne0.
\]
Eliminating $x_J$ with the help of (\ref{IVgen}), we obtain the identity
\[
 a^{iKj}+a^{jIk}-\mu(p^{IK}+p^{Ki}+p^{ik}+x_j)+r^{ki}=0\quad\Rightarrow\quad
 p^{IK}_{IK}=0.
\]
The second equation is shown analogously, by comparing (\ref{IVgen}) with the tripodal form $(I,j,K)$ (it is enough to use the symmetry $i\leftrightarrow K$, $k\leftrightarrow I$).

Thus, we have shown that (\ref{IVgen}) takes one of the following two forms:
\[
 p^I+p^{Ki}+p^k+x_j+x_J=0\qquad \text{or}\qquad p^{IK}+p^{ik}+x_j+x_J=0.
\]
But it is clear that the first case is included into the second one, upon the re-labeling $i\to k\to I\to K\to i$. Therefore we can assume that in equation (\ref{IVgen}) there holds $p^{Ki}=0$. Note that this makes the equation invariant under the re-labeling  $i\leftrightarrow k$, $I\leftrightarrow K$, which does not affect the tripodal form $(i,j,k)$. Now, returning to identity (\ref{IV-id}), we find:
\[
 p^{ik}_i=\frac{a^{iKj}_i+a^{kJi}_i}{a^{kJi}_J}\quad\Rightarrow\quad
 a^{iKj}_{iK}=a^{iKj}_{ij}=0
\]
and, using the symmetry just pointed out,
\[
 p^{ik}_k=\frac{a^{jIk}_k+a^{kJi}_k}{a^{kJi}_J}\quad\Rightarrow\quad
 a^{jIk}_{Ik}=a^{jIk}_{jk}=0.
\]
As a result, we put the tripodal forms $(i,j,K)$ and $(i,j,k)$ as
\[
 p^{IK}+p^{ik}+x_j+x_J=0,\qquad a^{Kj}+a^{jI}+a^{kJi}=0,
\]
and, solving for $x_J$, we arrive at the identity
\[
 p^{IK}+x_j=\varphi(x_k,x_i,a^{Kj}+a^{jI}) \quad\Rightarrow\quad
 \frac{p^{IK}_K}{p^{IK}_I}=\frac{a^{Kj}_K}{a^{jI}_I}=\frac{\a^K}{\a^I},
\]
where the last equality is obtained by setting $x_j=\const$. Upon suitable point transformations of $x_I$ and $x_k$, we can put $p^{IK}$ as $p^{IK}=P(x_I+x_K)$, while $a^{Kj}_K=a^{jI}_I=r^j$. There follows easily
\[
 \frac1{P'}=\frac{a^{Kj}_j+a^{jI}_j}{r^j},
\]
and, differentiating twice with respect to $x_I$ or to $x_K$, we find $(1/P')''=0$, or
\[
 P(y)=\a\log(y+\b)+\g \qquad\text{or}\qquad P(y)=\a y+\b.
\]
Of course, for $p^{ik}$ analogous formulas are found (by comparison with the tripodal form $(I,j,K)$ and upon suitable point transformations of $x_i,x_k$), so $p^{ik}=Q(x_i+x_k)$ with
\[
 Q(y)=\la\log(y+\mu)+\nu \qquad\text{or}\qquad Q(y)=\la y+\mu.
\]
Different combinations lead, after further obvious point transformations, to the cases (\ref{Y4}), (\ref{Y6}), (\ref{Y7}).
\end{proof}

\subsection{Class II}\label{ss:caseII}

Now let $a^{ikj}_{ikj}=0$ for all $i,k,j$, but $a^{ikj}_{ij}\ne0$ for at least one triple. For this triple we find two equivalent tripodal forms
(\ref{ikKj-II})
\[
\begin{gathered}
 (a^{ij}+p^{ik}+p^{kj})+(c^{jI}+p^{Kj})+(c^{Ji}+p^{iK})=0\\
 \Leftrightarrow\qquad
 (a^{ij}+p^{iK}+p^{Kj})+(c^{jI}+p^{kj})+(c^{Ji}+p^{ik})=0
\end{gathered}
\]
which have to be compared with other ones. Here it is convenient to distinguish subcases depending on which mixed derivatives of the pairs of functions $p^{ik},p^{kj}$ and $p^{iK},p^{Kj}$ vanish.

First assume that for at least one pair the mixed derivatives of the both functions do not vanish. For the sake of definiteness, let $p^{ik}_{ik}\ne0$, $p^{kj}_{kj}\ne0$. Note that formula (\ref{baba}) with constant $b^k,b^K$ yields that
\[
 a^{ikj}_{ij}\ne0\quad\Rightarrow\quad
 a^{jIK}_{IK}=a^{KJi}_{KJ}=a^{jIk}_{Ik}=a^{kJi}_{kJ}=0.
\]
Applying this implication to the second of the tripodal forms, we find:
 \begin{align*}
  & a^{jIk}_{jIk}=0,~~a^{jIk}_{jk}=p^{kj}_{kj}\ne0\quad\Rightarrow\quad
    a^{kJi}_{Ji}=a^{iKj}_{iK}=0\quad\Rightarrow\quad c^{Ji}_{Ji}=p^{iK}_{iK}=0;\\
  & a^{kJi}_{kJi}=0,~~a^{kJi}_{ki}=p^{ik}_{ik}\ne0\quad\Rightarrow\quad
    a^{jIk}_{jI}=a^{iKj}_{Kj}=0\quad\Rightarrow\quad c^{jI}_{jI}=p^{Kj}_{Kj}=0,
 \end{align*}
and as a consequence the equation takes the form
\[
 a^{ij}+p^{ik}+p^{kj}+q^I+q^J+q^K=0,
\]
and after the change $q^n\to x_n$ we come to equation
\begin{equation}\label{spec_2}
 p^{ij}+p^{jk}+p^{ki}+x_I+x_J+x_K=0.
\end{equation}

If the previous subcase does not take place, then we assume first that one of the pairs
$p^{ik},p^{kj}$ or $p^{iK},p^{Kj}$ contains exactly one function with vanishing mixed derivatives. For definiteness, let $p^{ik}_{ik}\ne0$ and $p^{kj}_{kj}=0$. Then, as before, $c^{jI}_{jI}=p^{Kj}_{Kj}=0$, and we come to an equation of the form
\[
 a^{ij}+p^{ik}+q^I+c^{Ji}+p^{iK}=0,
\]
or, upon a point change of variables, to equation
\begin{equation}\label{spec_3}
 p^{ij}+p^{ik}+p^{iJ}+p^{iK}+x_I=0.
\end{equation}

Finally, if all mixed derivatives $p^{ik}_{ik},p^{kj}_{kj},p^{iK}_{iK},p^{Kj}_{Kj}$ vanish, then our equation belongs to the special type already dealt with in Proposition \ref{st:spec_1}:
\[
 a^{ij}+q^k+c^{jI}+q^K+c^{Ji}=0;
\]
we have seen that such equations cannot belong to the case II.

As a result, case II is reduced to equations of two special types. They are dealt with in the following two statements which result in equation (\ref{Y5}) as the only possible one.

\begin{statement}\label{st:spec_2}
There exist no equations of the class II with one of the tripodal forms as in (\ref{spec_2}).
\end{statement}
\begin{proof}
We compare (\ref{spec_2}) with the tripodal form
\[
 (I,J,K):\quad a^{IkJ}+a^{JiK}+a^{KjI}=0.
\]
Solving for $x_k$, we obtain an identity
\[
 f(x_I,x_J,a^{JiK}+a^{KjI})=g(x_i,x_j,x_I+x_J+x_K),
\]
whence
\[
 \log a^{JiK}_i-\log a^{KjI}_j=\log g_i-\log g_j=h(x_i,x_j,x_I+x_J+x_K).
\]
Differentiating with respect to $x_I$, $x_J$, and $x_K$, and using the symmetry of the equation under permutations of $i,j,k$, we find:
\begin{alignat*}{4}
  &-\frac{a^{KjI}_{jI}}{a^{KjI}_j}
  &~~=~~&\frac{a^{JiK}_{Ji}}{a^{JiK}_i}
  &~~=~~&\frac{a^{JiK}_{iK}}{a^{JiK}_i}-\frac{a^{KjI}_{Kj}}{a^{KjI}_j}&~~=~~&\nu^K,\\
  &-\frac{a^{IkJ}_{kJ}}{a^{IkJ}_k}
  &~~=~~&\frac{a^{KjI}_{Kj}}{a^{KjI}_j}
  &~~=~~&\frac{a^{KjI}_{jI}}{a^{KjI}_j}-\frac{a^{IkJ}_{Ik}}{a^{IkJ}_k}&~~=~~&\la^I,\\
  &-\frac{a^{JiK}_{iK}}{a^{JiK}_i}
  &~~=~~&\frac{a^{IkJ}_{Ik}}{a^{IkJ}_k}
  &~~=~~&\frac{a^{IkJ}_{kJ}}{a^{IkJ}_k}-\frac{a^{JiK}_{Ji}}{a^{JiK}_i}&~~=~~&\mu^J,
 \end{alignat*}
where $\la^I$, $\mu^J$, and $\nu^K$ denote the common values of the corresponding expressions. Since, by assumption, the equation belongs to the class II, we have $a^{JiK}_{JiK}=0$, and, differentiating the last equality in the first equation with respect to  $x_J$, we find:
\[
 \mu^J\nu^K=-\frac{a^{JiK}_{iK}a^{JiK}_{Ji}}{(a^{JiK}_i)^2}=0,
\]
and analogously $\nu^K\la^I=\la^I\mu^J=0$; since there holds additionally
$\la^I+\mu^J+\nu^K=0$, all three functions actually vanish. This means that the tripodal form $(IJK)$ is reduced, upon a point transformation of $x_i,x_j,x_k$, to the form analogous to (\ref{spec_2}):
\[
 a^{IJ}+a^{JK}+a^{KI}+x_i+x_j+x_k=0.
\]
Resolving once again the both tripodal forms for $x_k$, we find an identity
\[
 a^{IJ}+a^{JK}+a^{KI}+x_i+x_j=g(x_i,x_j,x_I+x_J+x_K),
\]
whence
\[
 a^{IJ}_I+a^{KI}_I=a^{IJ}_J+a^{JK}_J=a^{JK}_K+a^{KI}_K.
\]
These equations are easy to solve, and we arrive at the tripodal form
\[
 \la(x_I+x_J+x_K)^2+\mu(x_I+x_J+x_K)+\nu+x_i+x_j+x_k=0.
\]
The case $\la=0$ leads to equation (\ref{Y7}) which does not belong to the class II, while for $\la\ne0$, as easily shown, this tripodal form cannot be equivalent to equation (\ref{spec_2}).
\end{proof}

\begin{statement}\label{st:spec_3}
Equations of the class II with one of the tripodal forms as in (\ref{spec_3}) can be reduced to (\ref{Y5}).
\end{statement}
\begin{proof}
Comparing with the tripodal form
\[
 (I,J,K):\quad a^{IkJ}+a^{JiK}+a^{KjI}=0,
\]
we find:
\begin{equation}\label{pp}
 \frac{p^{iJ}_J}{p^{ik}_k}=\frac{a^{JiK}_J+a^{IkJ}_J}{a^{IkJ}_k},\qquad
 \frac{p^{iK}_K}{p^{ij}_j}=\frac{a^{JiK}_K+a^{KjI}_K}{a^{KjI}_j}.
\end{equation}
There follows $a^{JiK}_{JK}=0$. Further, assume that $a^{JiK}_{Ji}\ne0$. Then, differentiating the first equation with respect to $x_I$ and $x_i$, we find $a^{IkJ}_{Ik}=a^{IkJ}_{IJ}=0$, so that the tripodal form is
\[
 a^{kJ}+a^{JiK}+a^{KjI}=0.
\]
Solving for $x_I$, we find:
\begin{equation}\label{pppp}
 p^{ij}+p^{ik}+p^{iJ}+p^{iK}=f(x_K,x_j,a^{kJ}+a^{JiK}),
\end{equation}
and applying $\partial_j\partial_k$ leads to $0=f'_j$, where the prime stands for the derivative of $f$ with respect to the third argument. But then  $p^{ij}_{ij}=0$.

Exchanging the roles of $j$ and $k$, we prove that the implications hold:
\[
 a^{JiK}_{Ji}\ne0~\Rightarrow~a^{IkJ}_{Ik}=a^{IkJ}_{IJ}=p^{ij}_{ij}=0;\quad
 a^{JiK}_{iK}\ne0~\Rightarrow~a^{KjI}_{jI}=a^{KjI}_{KI}=p^{ik}_{ik}=0.
\]
If $a^{JiK}_{Ji}\ne0$ and $a^{JiK}_{iK}\ne0$ simultaneously, then $p^{ij}_{ij}=p^{ik}_{ik}=0$, and our equation reduces to
\[
 p^j+p^k+p^{iJ}+p^{iK}+x_I=0,
\]
but this is, up to notation, a particular case of the previous special equation (\ref{spec_2}). Let, for example, $a^{JiK}_{Ji}\ne0$ and $a^{JiK}_{iK}=0$. Then (\ref{pppp}) is fulfilled, and, applying  $\partial_J\partial_K$, we come to $0=f'_K$ (recall that $a^{JiK}_{JK}=0$), but then also $p^{iK}_{iK}=0$, and we come again to an equation of the type (\ref{spec_2}):
\[
 p^j+p^{ik}+p^{iJ}+p^K+x_I=0.
\]

There remains the only possibility $a^{JiK}_{Ji}=a^{JiK}_{iK}=0$. In this case there follows from  (\ref{pp}) that
\[
  \left(\frac{p^{iJ}_J}{p^{ik\vphantom{j}}_k}\right)_{\!i}
 =\left(\frac{p^{iK}_K}{p^{ij}_j}\right)_{\!i}=0,
\]
and, using the symmetry of our equation under $k\leftrightarrow K$ (that is, comparing (\ref{spec_3}) with the tripodal form $(IJk)$ instead of $(IJK)$), we find also:
\[
  \left(\frac{p^{iJ}_J}{p^{iK\vphantom{j}}_K}\right)_{\!i}
 =\left(\frac{p^{ik}_k}{p^{ij}_j}\right)_{\!i}=0.
\]
All these equations yield:
\[
 p^{ij}+p^{ik}+p^{iJ}+p^{iK}=(q^j+q^k+q^J+q^K)r^i+s^i,
\]
moreover $r_i\ne\const$, since otherwise we would come to equation (\ref{Y7}) which does not belong to the class II. Up to changes of variables, equation (\ref{spec_3}) is reduced to
\[
 x_j+x_k+x_J+x_K=x_Ix_i+s^i.
\]
The tripodal form  $(I,J,K)$ in this case is $a^{IkJ}+a^i+a^{KjI}=0$, and, solving for  $x_j$, we find:
\[
x_Ix_i+s^i-x_k-x_J-x_K=f(x_K,x_I,a^{IkJ}+a^i)\quad\Rightarrow\quad
 x_I+s^i_i=-\frac{a^i_i}{a^{IkJ}_k}.
\]
Differentiating the last equation with respect to $x_i$ and $x_I$, we obtain: $a^i_{ii}=0$, but then
$s^i_{ii}=0$, and an additional linear change of variables leads to equation (\ref{Y5}).
\end{proof}

\subsection{Class III}\label{ss:caseIII}

In this case it turns out to be convenient to change notation and to re-write the tripodal forms (\ref{3legs}) as
\begin{equation}\label{6legs}
\begin{array}{ccc}
 \text{\em faces:} &\qquad & \text{\em equations:} \\[1.5mm]
  (1,2,3),~(4,5,6) && p^{14}+p^{42}+p^{26}+p^{63}+p^{35}+p^{51}=0, \\
  (1,2,4),~(3,5,6) && q^{13}+q^{32}+q^{26}+q^{64}+q^{45}+q^{51}=0, \\
  (1,3,5),~(2,4,6) && r^{63}+r^{32}+r^{21}+r^{14}+r^{45}+r^{56}=0, \\
  (1,4,5),~(2,3,6) && s^{64}+s^{42}+s^{21}+s^{13}+s^{35}+s^{56}=0
\end{array}
\end{equation}
(through a re-ordering of terms each of these equations gives two tripodal forms, corresponding to opposite faces).

\begin{statement}
Each term in equations (\ref{6legs}), as a function $f(x,y)$ of its arguments, is of the following type:
\begin{equation}\label{fabcd}
 f=a(x)b(y)+c(x)+d(y)\quad\text{or}\quad f=\rho\log(a(x)+b(y))+c(x)+d(y).
\end{equation}
\end{statement}
\begin{proof}
First we show that
\begin{equation}\label{fxy}
 f_x=\frac{\a(x)}{\b(x)+\d(y)}+\g(x),\qquad
 f_y=\frac{\la(y)}{\mu(y)+\varkappa(x)}+\nu(y).
\end{equation}
Due to the symmetry of formulas (\ref{6legs}) it is enough to do this for $f=p^{14}$, $x=x_1$, $y=x_4$. We use the first identity (\ref{id1}), with $(i,j,k)=(1,2,3)$ and $(4,5,6)$:
\[
 \frac{a^{132}_1+a^{451}_1}{a^{451}_5}=\frac{a^{142}_1+a^{351}_1}{a^{351}_5},\qquad
 \frac{a^{465}_4+a^{124}_4}{a^{124}_2}=\frac{a^{415}_4+a^{624}_4}{a^{624}_2}.
\]
This yields:
\begin{equation}\label{id1'}
 \frac{q^{13}_1+q^{51}_1}{q^{45}_5+q^{51}_5}=
 \frac{p^{14}_1+p^{51}_1}{p^{35}_5+p^{51}_5},\qquad
 \frac{s^{64}_4+s^{42}_4}{s^{21}_2+s^{42}_2}=
 \frac{p^{14}_4+p^{42}_4}{p^{42}_2+p^{26}_2},
\end{equation}
and, setting all variables except for $x_1,x_4$ to constants, we come to (\ref{fxy}).

Further, computing the mixed derivatives, we find:
\[
 -f_{xy}=\frac{\a(x)\d'(y)}{(\b(x)+\d(y))^2}=
 \frac{\la(y)\varkappa'(x)}{(\mu(y)+\varkappa(x))^2}.
\]
If at least one of the functions $\a,\d',\la,\varkappa'$ vanishes, then $f=c(x)+d(y)$. Otherwise, taking logarithms and applying $\partial_x\partial_y$, we come to
\[
 \frac{\b'(x)\d'(y)}{(\b(x)+\d(y))^2}=
 \frac{\mu'(y)\varkappa'(x)}{(\mu(y)+\varkappa(x))^2}.
\]
If $\b'=\mu'=0$, then equations (\ref{fxy}) are easily integrated and lead to $f=a(x)b(y)+c(x)+d(y)$. Finally, if $\b'$ and $\mu'$ do not vanish, then the two latter equations yield $\a/\b'=\la/\mu'=\const$, and an integration finishes the proof.
\end{proof}

The proposition just proven is not easy to use directly, since it is not yet known how are the representations (\ref{fabcd}) corresponding to different terms of the tripodal forms (\ref{6legs}) related to one another. We start with filtering away the case when these forms contain too many terms of the kind $a(x)+b(y)$.

\begin{statement}
If each of the forms (\ref{6legs}) contains terms with vanishing mixed derivatives, then one of the forms contains at least four such terms, and it is a particular case of the equation (\ref{spec_1}).
\end{statement}
\begin{proof}
By differentiation we derive from (\ref{id1'}):
\[
 s^{21}_{21}=0~~\Leftrightarrow~~p^{14}_{14}=0~~\Leftrightarrow~~q^{45}_{45}=0,
\]
that is, the presence of one term with vanishing mixed derivatives in one of the forms yields the presence of such terms in further two forms. This chain can be continued to a closed cycle:
\[
 p^{14}_{14}=0~\Rightarrow~q^{45}_{45}=0~\Rightarrow~
 s^{56}_{56}=0~\Rightarrow~p^{63}_{63}=0~\Rightarrow~
 q^{32}_{32}=0~\Rightarrow~s^{21}_{21}=0~\Rightarrow~p^{14}_{14}=0,
\]
not containing $r^{ij}$. Of course, all forms are on the same footing, so that starting, e.g., from $r^{14}$, we would get the cycle
\[
 r^{14}_{14}=0~\Rightarrow~p^{42}_{42}=0~\Rightarrow~
 q^{26}_{26}=0~\Rightarrow~r^{63}_{63}=0~\Rightarrow~
 p^{35}_{35}=0~\Rightarrow~q^{51}_{51}=0~\Rightarrow~r^{14}_{14}=0,
\]
not containing $s^{ij}$. Therefore, if the terms with vanishing mixed derivatives are present in each of the forms, then one of the forms accumulates at least four such terms. For instance, if both cycles above take place then $p^{14}_{14}=p^{63}_{63}=p^{42}_{42}=p^{35}_{35}=0$, and the first equation in (\ref{6legs}) becomes $p^4+p^{26}+p^3+p^{51}=0$.
\end{proof}

To compare the general terms (\ref{fabcd}) we will use the following lemma. Note that the function $F$ in this lemma defines a general solution of the Liouville equation $(\log F)_{xy}=-2F$.

\begin{lemma}\label{l:frac}
All solutions of the functional equation
\[
 \frac{g_x(x,z)h_y(y,z)}{(g(x,z)-h(y,z))^2}=F(x,y)\ne0
\]
are given by the formulas
\[
 g=\frac{\d(z)}{\g(z)-\a(x)}+\eps(z),\quad
 h=\frac{\d(z)}{\g(z)-\b(y)}+\eps(z),\quad
 F=\frac{\a'(x)\b'(y)}{(\a(x)-\b(y))^2}.
\]
\end{lemma}
\begin{proof}
Integration with respect to $x$ leads to
\[
 \frac{h_y(y,z)}{g(x,z)-h(y,z)}=\tilde h(y,z)-\int F(x,y)dx.
\]
Setting $y=y_0=\const$, we come to
\[
 \frac{\d(z)}{g(x,z)-\eps(z)}=\g(z)-\a(x),
\]
and since $h_y\ne0$, we can choose $y_0$ so that the numerator and the denominator do not vanish identically. This yields the expression for $g$. Further, we can use the invariance of the equation under non-degenerate linear-fractional transformations
\[
 g\to\frac{p(z)g+q(z)}{r(z)g+s(z)},\qquad h\to\frac{p(z)h+q(z)}{r(z)h+s(z)},
 \qquad F\to F
\]
to bring $g$ to the form $g=\a(x)$, $\a'\ne0$. Then
\[
 \frac{\a'(x)h_y(y,z)}{(\a(x)-h(y,z))^2}=F(x,y)\quad\Rightarrow\quad
 \frac{h_{yz}}{h_y}+\frac{2h_z}{h-\a(x)}=0\quad\Rightarrow\quad
 \frac{h_z\a'}{(h-\a(x))^2}=0,
\]
that is, $h=\b(y)$, and the inverse linear-fractional transformation finishes the proof.
\end{proof}

Now we are in a position to analyze the case where at least one of the tripodal forms (\ref{6legs}) contains no terms with the vanishing mixed derivatives (representable as a sum of functions of a single variable). For definiteness, let it be the first form in (\ref{6legs}).

\begin{statement}\label{st:T23}
If $p^{ij}_{ij}\ne0$ for all $i,j$, the equation (\ref{xxxxxx}) is equivalent to (\ref{Y3}) or to
(\ref{Y2}).
\end{statement}
\begin{proof}
We apply Proposition \ref{st:imp2}. We have: $a^{ikj}=p^{ik}+p^{kj}$, $a^{jIK}=p^{jI}+p^{IK}$, $a^{KJi}=p^{KJ}+p^{Ji}$, and, since by assumption $a^{jIK}_{IK}=p^{IK}_{IK}\ne0$, the following equation is fulfilled:
\[
 \frac{p^{ik}_{ik}p^{kj}_{kj}}{(p^{ik}_k+p^{kj}_k)^2}=
 \frac{a^{iKj}_{iK}a^{iKj}_{Kj}}{(a^{iKj}_K)^2}=A^{ij}.
\]
Then, according to Lemma  \ref{l:frac}, we have:
\[
 p^{ik}_k+p^{kj}_k=\frac{\d^k}{\g^k-\a^i}-\frac{\d^k}{\g^k-\b^j},
\]
so that the representations (\ref{fabcd}) for $p^{ik}$ and for $p^{kj}$ are tied in the following sense:\medskip

if $p^{ik}=a^ib^k+c^i+d^k$, then $p^{kj}=a^jb^k+c^j-d^k$;
\medskip

if $p^{ik}=\rho\log(a^i+b^k)+c^i+d^k$, then $p^{kj}=-\rho\log(a^j+b^k)+c^j-d^k$.
\medskip

\noindent Applying this to the pair $p^{kj}$, $p^{jI}$, and further cyclically, we obtain the representations
\[
 \begin{array}{l}
  p^{ik}=a^ib^k+c^i+d^k,\\
  p^{kj}=a^jb^k+c^j-d^k,\\
  p^{jI}=a^jb^I-c^j+d^I,\\
  p^{IK}=a^Kb^I+c^K-d^I,\\
  p^{KJ}=a^Kb^J-c^K+d^J,\\
  p^{Ji}=\hat a^ib^J+\hat c^i-d^J,\\
 \end{array}
\qquad\text{or}\qquad
 \begin{array}{l}
  p^{ik}=\rho\log(a^i+b^k)+c^i+d^k,\\
  p^{kj}=-\rho\log(a^j+b^k)+c^j-d^k,\\
  p^{jI}=\rho\log(a^j+b^I)-c^j+d^I,\\
  p^{IK}=-\rho\log(a^K+b^I)+c^K-d^I,\\
  p^{KJ}=\rho\log(a^K+b^J)-c^K+d^J,\\
  p^{Ji}=-\rho\log(\tilde a^i+b^J)+\tilde c^i-d^J.\\
 \end{array}
\]
A comparison of the first and the last terms shows that
\[
 \hat a^i=\eps a^i,\quad \hat c^i=-c^i+\d,\qquad
 \tilde a^i=\eps a^i+\d,\quad \tilde c^i=-c^i+\sigma.
\]
Summing up all the forms, we obtain, after point transformations and re-enumeration, the following equations:
\begin{gather*}
 x_1x_4+x_4x_2+x_2x_6+x_6x_3+x_3x_5+\eps x_5x_1=\d,\\
 (x_1+x_4)(x_2+x_6)(x_3+x_5)=\sigma(x_4+x_2)(x_6+x_3)(x_5+\eps x_1+\d).
\end{gather*}
One can be more precise about the constants here. For this aim, we solve these equations for one of the variables, say for $x_6=f(x_1,x_2,x_3,x_4,x_5)$, and compare with the second tripodal form (\ref{6legs}), which yields another expression for
\[
 x_6=\varphi(x_2,x_4,q^{13}+q^{32}+q^{45}+q^{51}).
\]
It follows that $(\log(f_3/f_5))_{24}=0$, which can be checked in a straightforward way for a given function $f$. This check shows that the constants are uniquely determined, and we come to the equations
\begin{gather*}
 x_1x_4+x_4x_2+x_2x_6+x_6x_3+x_3x_5-x_5x_1=0,\\
 (x_1+x_4)(x_2+x_6)(x_3+x_5)=(x_4+x_2)(x_6+x_3)(x_5+x_1).
\end{gather*}
The first of them is reduced to (\ref{Y2}) after an additional change $x_2\to-x_2$, $x_6\to-x_6$,
while the second one coincides with (\ref{Y3}) after changing the signs of $x_4,x_5,x_6$.
\end{proof}

This finishes the proof of Theorem \ref{th:3leg_types}.

\section{Classification of compatible quintuples}\label{s:class}

\subsection{Separating away non-compatible equations}\label{ss:killbill}

Now we have to combine tripodal equations into compatible quintuples, taking into account that the legs on the faces shared by different octahedra must coincide, according to Proposition \ref{st:3leg}. Table \ref{fig:sklad} allows us to exclude some a priori non-compatible combinations. For instance, equation of the type (\ref{Y1}) can only match either an equation of the same type or an equation of the type (\ref{Y6}). Other combinations are impossible, as shown in the following proposition.

\begin{statement}
The functions
\begin{gather*}
 xyz,\quad xy,\quad y,\quad (x+z)^ky\;(k\ne0),\quad (x+y)z,\\
 y+\log(x+z),\quad \log(x+y)-\log(y+z),\quad \log(x+y),
\end{gather*}
are pairwise non-equivalent modulo transformations
\[
  \tilde a(x,y,z)=\g a(f(x),g(y),h(z))+\mu(x)+\nu(z)
\]
with non-constant $f,g,h$, along with the flip $x\leftrightarrow z$. Functions $(x+z)^ky$ with different $k$ are also non-equivalent.
\end{statement}
\begin{proof}
All proofs are similar, therefore we consider several pairs of functions as examples.
\begin{itemize}
\item The equality
\[
 \g g(y)+\g\log(f(x)+h(z))+\mu(x)+\nu(z)=\log\frac{x+y}{y+z},
\]
is impossible, as follows immediately by differentiation with respect to $x$ and $z$.

\item Suppose that
\[
 \g f(x)g(y)h(z)+\mu(x)+\nu(z)=(x+y)z,
\]
then the differentiation gives $\g f'(x)g'(y)h'(z)=0$.

\item Suppose that
\[
 \g(f(x)+h(z))^kg(y)+\mu(x)+\nu(z)=(x+z)^my.
\]
Then $g(y)=\a y+\b$, and further $\a\g(f(x)+h(z))^k=(x+z)^m$, so
$f(x)+h(z)=\const(x+z)^{m/k}$. Differentiation with respect to $x,z$ yields $m=k$.
\end{itemize}
The reader is invited to work out the other pairs of functions.
\end{proof}

A more careful analysis shows that equations (\ref{Y1}) and (\ref{Y6}) cannot be matched, although they have a common leg. More precisely, the following statement holds true.

\begin{statement}\label{st:NoT456}
There do not exist compatible triples of equations where one of the equations is of the type
(\ref{Y4}) with $\g\ne1$, or (\ref{Y5}), or (\ref{Y6}).
\end{statement}
\begin{proof}
The proof is based on the fact that each of the equations listed in the proposition possesses a unique leg, which does not appear by equations of other types ($y(x+z)^\g$, $(x+y)z$, and $y+\log(x+z)$, respectively). By virtue of Proposition \ref{st:3leg}, there follows that at least one further equation of the compatible triple has to be of the same type. Since all coordinate directions $\Integer^4$ are on the same footing, we can assume that the compatible triple is given by equations (\ref{3leg:1}), the unique leg being the function $a(x_1,x_4,x_{14})$ in the first and the second equations. Moreover, since this leg is unique, the equations themselves are recovered up to transformations (\ref{tix}), in other words, the functions  $b(x_2,x_4,x_{24})$ and $c(x_3,x_4,x_{34})$ can be determined. It turns out that $b-c$ does not depend on  $x_4$, so that the third equation in (\ref{3leg:1}) is reducible, in contradiction with our standing assumption. In the following detailed analysis $X_i$ stands for an arbitrary non-constant function of $x_i$.

{\em Equation (\ref{Y4}) with $\g\ne1$}. Consider equation
\[
 X_1X_{24}=(X_2+X_{12})(X_4+X_{14})^{-\g}.
\]
Its tripodal form with the head $(4,14,24)$ is
\[
 X_1(X_4+X_{14})^\g-X_2/X_{24}=X_{12}/X_{24},
\]
and, comparing with the first equation in (\ref{3leg:1}), we find:
\begin{align*}
 a(x_1,x_4,x_{14})&=X_1(X_4+X_{14})^\g+\la(x_{14})+\mu(x_4),\\
 b(x_2,x_4,x_{24})&=X_2/X_{24}+\mu(x_4)+\nu(x_{24}).
\end{align*}
This $a$ corresponds to a unique leg $y(x+z)^\g$ in Table \ref{fig:sklad}, therefore the second equation in (\ref{3leg:1})
\[
 X_3/X_{34}-X_1(X_4+X_{14})^\g=-X_{13}/X_{34}.
\]
Therefore $c(x_3,x_4,x_{34})=X_3/X_{34}+\mu(x_4)+\kappa(x_{34})$, but then the terms $\mu(x_4)$ in the third equation in (\ref{3leg:1}) cancel.

{\em Equation (\ref{Y5})}. Consider equation
\[
 X_2X_{14}=X_1+X_4+X_{24}+X_{12} \quad\Rightarrow\quad
 \frac{X_1+X_4}{X_{14}}-X_2=-\frac{X_{24}+X_{12}}{X_{14}}.
\]
We have
\[
 a(x_1,x_4,x_{14})=\frac{X_1+X_4}{X_{14}}+\la(x_{14})+\mu(x_4),\quad
 b(x_2,x_4,x_{24})=X_2+\mu(x_4)+\nu(x_{24}).
\]
Due to the uniqueness of the leg $a$ the second equation in (\ref{3leg:1}) has to be of the type (\ref{Y5}), as well, and to have the tripodal form
\[
 X_3-\frac{X_1+X_4}{X_{14}}=\frac{X_{34}+X_{23}}{X_{14}},
\]
whence $c(x_3,x_4,x_{34})=X_3+\mu(x_4)+\kappa(x_{34})$, and the third equation in (\ref{3leg:1}) does not contain $x_4$.

{\em Equation (\ref{Y6})}. Consider equation
\[
 X_1X_2X_{12}X_{24}=X_4+X_{14} \quad\Rightarrow\quad
 (\log(X_4+X_{14})-\log X_1)-\log X_2=\log(X_{12}X_{24}).
\]
We have
\begin{align*}
 a(x_1,x_4,x_{14})&=\log(X_4+X_{14})-\log X_1+\la(x_{14})+\mu(x_4),\\
 b(x_2,x_4,x_{24})&=\log X_2+\mu(x_4)+\nu(x_{24}).
\end{align*}
Also this $a$ is unique, therefore the second equation in (\ref{3leg:1}) has to be of the type (\ref{Y6}) and to have the tripodal form
\[
 \log X_3-(\log(X_4+X_{14})-\log X_1)=-\log(X_{13}X_{34}),
\]
whence $c(x_3,x_4,x_{34})=\log X_3+\mu(x_4)+\kappa(x_{34})$, and, as in the previous case, the third equation in (\ref{3leg:1}) does not contain $x_4$.
\end{proof}

\subsection{Completing the classification}\label{ss:final}

Note that after removing equations (\ref{Y5}) and (\ref{Y6}) from Table \ref{fig:sklad} the leg $xy$ becomes unique, however the argumentation like in the proof of Proposition \ref{st:NoT456} is not possible. To handle with the remaining cases, it is not enough to use formulas (\ref{3leg:1}) alone, and one has to refer to all tripodal forms, (\ref{ijk}) and their shifted versions.

Before we formulate the final result, let us describe more precisely what transformations are allowed to bring the equations to the canonical form. First of all, these are autonomous point transformations $x\to X(x)$ (the same at all lattice points). However, if only these are allowed, then the answer will contain many arbitrary constant parameters. It turns out that all these parameters are inessential and can be killed by non-autonomous transformations, which depend on the lattice point. Generally speaking, such transformations result in non-autonomous equations, therefore not all them should be allowed but only those special ones which render the transformed equation still autonomous. Their existence is related to a certain symmetry of the equations. For instance, if the equation is invariant under the one-parameter group of translations $x\to x+a$, then it admits non-autonomous transformations $x(i,j,k)\to
x(i,j,k)+\a i+\b j +\g k$, with the transformed equation being dependent on the parameters $\a,\b,\g$, which can be used to simplify the result. Similar transformations exist in all cases under consideration, and we use them to eliminate all the constant parameters (the situation is similar for continuous 3D integrable systems).

\begin{theorem}\label{th:class}
Any compatible quintuple of irreducible shift-invariant octahedron type equations on $Q(A_4)$ is equivalent, modulo non-autonomous point transformations, to one of the following systems (different indices stand for the shifts in different coordinate directions; recall that in the $\mathbb Z^4$ realization of the lattice $Q(A_4)$ the shift $T_0$ in the coordinate direction 0 can be simply omitted).\smallskip

\noindent Five equations of type {\em (\ref{Y1})}:
\begin{equation}\label{5T1}\tag{$\chi_1^5$}
 x_{ij}x_{km}-x_{ik}x_{jm}+x_{im}x_{jk}=0,\qquad 0\le i<j<k<m\le 4;
\end{equation}
Five equations of type {\em(\ref{Y3})}:
\begin{equation}\label{5T3}\tag{$\chi_2^5$}
 \frac{(x_{ij}-x_{ik})(x_{kj}-x_{km})(x_{jm}-x_{im})}
      {(x_{ik}-x_{kj})(x_{km}-x_{jm})(x_{im}-x_{ij})}=-1,\qquad i,j,k,m\in\{0,1,2,3,4\};
\end{equation}
Two different quintuples consisting of four equations of type {\em(\ref{Y2})} and one equation of type {\em(\ref{Y3})}:
\begin{equation}\label{4T2-1T3}\tag{$\chi_3^4\chi_2$}
\left\{\begin{array}{l}
 (x_{ik}-x_{ij})x_{i0}+(x_{ij}-x_{jk})x_{j0}+(x_{jk}-x_{ik})x_{k0}=0,\qquad i,j,k\in\{1,2,3,4\},\vspace{1mm}\\
 \dfrac{(x_{12}-x_{13})(x_{23}-x_{34})(x_{24}-x_{14})}
       {(x_{13}-x_{23})(x_{34}-x_{24})(x_{13}-x_{12})}=-1;
\end{array}\right.
\end{equation}
and
\begin{equation}\label{4T2-1T3'}\tag{$\chi_4^4\chi_2$}
\left\{\begin{array}{l}
 \dfrac{x_{ik}-x_{ij}}{x_{i0}}+\dfrac{x_{ij}-x_{jk}}{x_{j0}}
 +\dfrac{x_{jk}-x_{ik}}{x_{k0}}=0,\qquad i,j,k\in\{1,2,3,4\},\vspace{1mm}\\
 \dfrac{(x_{12}-x_{13})(x_{23}-x_{34})(x_{24}-x_{14})}
       {(x_{13}-x_{23})(x_{34}-x_{24})(x_{13}-x_{12})}=-1;
\end{array}\right.
\end{equation}
Three equations of type {\em(\ref{Y4})} and two equations of type {\em(\ref{Y2})}:
\begin{equation}\label{3T4-2T2}\tag{$\chi_5^3\chi_4^2$}
\left\{\begin{array}{l}
 \dfrac{x_{i4}-x_{j4}}{x_{40}}
   =x_{ij}\Bigl(\dfrac1{x_{j0}}-\dfrac1{x_{i0}}\Bigr),\qquad i,j\in\{1,2,3\},\vspace{2mm}\\
 \dfrac{x_{13}-x_{12}}{x_{10}}+\dfrac{x_{12}-x_{23}}{x_{20}}
  +\dfrac{x_{23}-x_{13}}{x_{30}}=0,\vspace{2mm}\\
 \dfrac{x_{14}-x_{24}}{x_{12}}+\dfrac{x_{24}-x_{34}}{x_{23}}
  +\dfrac{x_{34}-x_{14}}{x_{13}}=0.
\end{array}\right.
\end{equation}
\end{theorem}
\begin{proof}
{\em General scheme.} We start with one tripodal equation, replacing variables $x_{ij}$ by yet unknown non-constant functions $X_{ij}=X_{ij}(x_{ij})$. Comparing terms in the tripodal forms (\ref{ijk}), we are able to completely determine the consistent quintuple, up to ten arbitrary functions $X_{12},\dots,X_{45}$. To specify the latter, we use the shifted tripodal forms. One can show then that all functions $X_{ij}$ are related to one another via some linear-fractional transformations. Point transformations allow us to assume that the functions $X_{ij}$ are linear-fractional, with coefficients of different functions being connected by certain relations. Resolving these relations, we come to a consistent system containing several free parameters. Finally, we get rid of them using non-autonomous transformations.

{\em Equation} (\ref{Y1}) can only be compatible with equations of the same type, as follows from Table \ref{fig:sklad}. Equations (\ref{ijk}) with $(i,j,k,m,n)=(1,2,3,4,0)$ read:
\begin{equation}\label{T1-ijk}
\begin{aligned}
 \EQ1:&\quad \frac{X_{34}}{X_{03}X_{04}}
            -\frac{X_{24}}{X_{02}X_{04}}
            +\frac{X_{23}}{X_{02}X_{03}}=0,\\
 \EQ2:&\quad \frac{X_{34}}{X_{03}X_{04}}
            -\frac{X_{14}}{X_{01}X_{04}}
            +\frac{X_{13}}{X_{01}X_{03}}=0,\\
 \EQ3:&\quad \frac{X_{24}}{X_{02}X_{04}}
            -\frac{X_{14}}{X_{01}X_{04}}
            +\frac{X_{12}}{X_{01}X_{02}}=0.
\end{aligned}
\end{equation}
Indeed, in the first equation the functions  $X_{ij}$ can be freely chosen; this defines the first term in the second equation. But then in the second term in $\EQ2$ the dependence on $x_{40}$ is known, while the variables $x_{14}$ and $x_{01}$ (which were absent from $\EQ1$) enter via arbitrary new functions. Continuing in this fashion, we express all terms through ten arbitrary functions $X_{ij}$. As a corollary we get
\[
 \EQ4:~~X_{01}X_{23}-X_{02}X_{13}+X_{03}X_{12}=0,\qquad
 \langle 0\rangle:~~X_{12}X_{34}-X_{13}X_{24}+X_{14}X_{23}=0.
\]
Equations (\ref{ijk}) are herewith exhausted. To determine the functions $X_{ij}$ we have to consider the shifted tripodal forms. We write them as
\begin{align*}
 T_1\EQ1:&\quad T_1\left(\frac{X_{02}}{X_{23}X_{24}}
         -\frac{X_{03}}{X_{23}X_{34}}
         +\frac{X_{04}}{X_{24}X_{34}}=0\right),\\
 T_2\EQ2:&\quad T_2\left(\frac{X_{01}}{X_{13}X_{14}}
         -\frac{X_{03}}{X_{13}X_{34}}
         +\frac{X_{04}}{X_{14}X_{34}}=0\right),\\
 T_3\EQ3:&\quad T_3\left(\frac{X_{01}}{X_{12}X_{14}}
         -\frac{X_{02}}{X_{12}X_{24}}
         +\frac{X_{04}}{X_{14}X_{24}}=0\right),
\end{align*}
and the comparison of the legs on the left-hand sides leads to relations like
\[
 \frac{\a X_{30}(x_0)}{X_{23}(x_2)X_{34}(x_4)}
 =\frac{X_{10}(x_0)}{X_{12}(x_2)X_{14}(x_4)}+\mu(x_2)+\nu(x_4).
\]
It is easy to realize that the terms $\mu+\nu$ can be neglected here, and that the functions $X_{12},X_{13},X_{23}$ coincide up to constant factors, and the same is true for the functions $X_{01},X_{02},X_{03}$. Because of the symmetry of all the indices, we come to the conclusion that all ten functions are proportional. Without loss of generality we can set $X_{ij}=\a_{ij}x_{ij}$, and a direct check shows that the equations are consistent for arbitrary values of $\a_{ij}$. All these parameters can be set equal to 1 upon use of the non-autonomous transformation
\[
 \tilde x(n_0,n_1,n_2,n_3,n_4)= \prod_{i,j}\a^{n_in_j}_{ij}x(n_0,n_1,n_2,n_3,n_4),
\]
and the answer takes the form of the quintuple (\ref{5T1}).\smallskip

{\em Equation} (\ref{Y4}) {\em with $\g=1$.} We consider this equation in the form
\[
 X_{04}(X_{24}-X_{34})=X_{23}(X_{02}-X_{03}).
\]
Identifying it with the first equation in (\ref{ijk}) with $(i,j,k,m,n)=(1,2,3,4,0)$, we have:
\[
 a^{02,24,04}=X_{24}X_{04},\qquad a^{03,34,04}=X_{34}X_{04},\qquad
 a^{02,23,03}=X_{23}(X_{03}-X_{02}),
\]
and the last two equations in (\ref{ijk}) read:
\[
 \EQ2:~~X_{34}X_{04}-a^{01,14,04}=a^{03,13,01},\qquad
 \EQ3:~~a^{01,14,04}-X_{24}X_{04}=a^{01,12,02}.
\]
Since the leg $xy$ is unique, these equations are of the type (\ref{Y4}), as well. This determines them in the following form:
\[
 X_{34}X_{04}-X_{14}X_{04}=X_{13}(X_{01}-Y_{03}),\qquad
 X_{14}X_{04}-X_{24}X_{04}=X_{12}(Y_{02}-Y_{01}),
\]
where $Y$, similarly to $X$, stand for arbitrary non-constant functions of the corresponding variables: $Y_{ij}=Y_{ij}(x_{ij})$. Summing up, we find the fourth equation
\[
 \EQ4:~~X_{12}(Y_{02}-Y_{01})+X_{13}(X_{01}-Y_{03})+X_{23}(X_{03}-X_{02})=0,
\]
which can be of the type (\ref{Y2}) only. But then the corresponding functions $X$ and $Y$ are linearly related, and it is easy to see that we can assume that they plainly coincide (re-defining $X_{12},X_{13}$, and $X_{15}$, if necessary). As a result, we come to the following quintuple of equations:
\begin{equation}\label{T4-ijk}
\begin{aligned}
 \EQ1:&\quad (X_{24}-X_{34})X_{04}=X_{23}(X_{03}-X_{02}),\\
 \EQ2:&\quad (X_{34}-X_{14})X_{04}=X_{13}(X_{01}-X_{03}),\\
 \EQ3:&\quad (X_{14}-X_{24})X_{04}=X_{12}(X_{02}-X_{01}),\\
 \EQ4:&\quad X_{12}(X_{02}-X_{01})+X_{13}(X_{01}-X_{03})+X_{23}(X_{03}-X_{02})=0,\\
 \EQ0:&\quad \frac{X_{14}-X_{24}}{X_{12}}+\frac{X_{24}-X_{34}}{X_{23}}
  +\frac{X_{34}-X_{14}}{X_{13}}=0.
\end{aligned}
\end{equation}
Equations (\ref{ijk}) are now exhausted, and the remaining functional freedom has to be reduced by using their shifted versions:
\begin{align*}
 T_1\EQ1: &\quad T_1\bigl(X_{23}X_{03}-X_{23}X_{02}-X_{04}(X_{24}-X_{34})=0\bigr),\\
 T_2\EQ2: &\quad T_2\bigl(X_{13}X_{01}-X_{13}X_{03}-X_{04}(X_{34}-X_{14})=0\bigr),\\
 T_3\EQ3: &\quad T_3\bigl(X_{12}X_{02}-X_{12}X_{01}-X_{04}(X_{14}-X_{24})=0\bigr).
\end{align*}
The comparison of legs on the left-hand sides leads to relations like
\[
 \a X_{23}(x_2)X_{03}(x_0)=X_{12}(x_2)X_{01}(x_0)+\mu(x_2)+\nu(x_4),
\]
where the terms $\mu+\nu$ account for the freedom in the definition of the legs. There follows that the functions $X_{12},X_{13},X_{23}$ are the same up to constant factors, while the relations between the functions $X_{15}$, $X_{25}$ and $X_{35}$ are affine. Further, equation $\EQ4$ takes the form
\begin{gather*}
  \a X_{04}(x_{01})(X_{24}(x_{12})-X_{34}(x_{13}))
  +\b X_{04}(x_{02})(X_{34}(x_{23})-X_{14}(x_{12}))\\
  +X_{04}(x_{03})(X_{14}(x_{13})-X_{24}(x_{23}))=0,
\end{gather*}
and the comparison with the previous form leads to the conclusion that $X_{14}$, $X_{24}$, and $X_{34}$ are affinely related to $X_{12}$, and $X_{04}$ is affinely related to $X_{15}$. Finally, for $\EQ0$ we obtain the representation
\[
  \frac{X_{02}(x_{12})-X_{03}(x_{13})}{X_{04}(x_{14})}
  +\g\frac{X_{03}(x_{23})-X_{01}(x_{12})}{X_{04}(x_{24})}
  +\d\frac{X_{01}(x_{13})-X_{02}(x_{23})}{X_{04}(x_{34})}=0
\]
(it is enough to observe that equations (\ref{T4-ijk}) are invariant under the flip $4\leftrightarrow 0$ and the inversion of $X_{04}$, $X_{12}$, $X_{13}$, $X_{23}$), and a comparison with the previous representation shows that $1/X_{04}$ is affinely related to $X_{14}$. As a result, all functions are expressed through one of them (say, $X_{12}$), and without loss of generality we can put them as functions of $x$ as
\begin{gather*}
 X_{12}=x,\quad X_{13}=px,\quad X_{23}=qx,\quad X_{04}=1/x,\\
 X_{i4}=a_ix+b_i,\quad X_{0i}=c_i/x+d_i,\quad i=1,2,3.
\end{gather*}
Substitution into the shifted tripodal forms allows us to show that $b_i=d_i=0$ and leads to equations
\[
 \frac{a_ix_{i4}-a_jx_{j4}}{x_{04}}
 =Ax_{ij}\Bigl(\frac{a_i}{x_{0j}}-\frac{a_j}{x_{0i}}\Bigr),\quad i,j=1,2,3.
\]
They are consistent for all parameter values. The non-autonomous change $\tilde x(i,j,k,m,n)=
a^i_1a^j_2a^k_3A^{mn}x(i,j,k,m,n)$ allows us to set all parameters to 1, and we come to the consistent triple (\ref{3T4-2T2}).\smallskip

{\em Equation} (\ref{Y2}). Now we do not consider equations of the type (\ref{Y4}) anymore, and the leg $y(x+z)$ becomes unique for equations of the type (\ref{Y2}). It is not difficult to see that in a consistent quintuple containing such an equation at least four equations are of the same type. Indeed, either the tripodal form does not contain legs $y(x+z)$, or all three legs are of this sort. Therefore, if one starts with such a leg for one equation of the triple (\ref{ijk}), then all the legs for the three equations and for their sum (which is the tripodal form of the fourth equation) will be of this sort. We easily come to the following system:
\begin{equation}\label{T2-ijk}
\begin{aligned}
 \EQ1:&\quad X_{24}(X_{02}-X_{04})+X_{34}(X_{04}-X_{03})+X_{23}(X_{03}-X_{02})=0,\\
 \EQ2:&\quad X_{34}(X_{03}-X_{04})+X_{14}(X_{04}-X_{01})+X_{13}(X_{01}-X_{03})=0,\\
 \EQ3:&\quad X_{14}(X_{01}-X_{04})+X_{24}(X_{04}-X_{02})+X_{12}(X_{02}-X_{01})=0,\\
 \EQ4:&\quad X_{12}(X_{02}-X_{01})+X_{13}(X_{01}-X_{03})+X_{23}(X_{03}-X_{02})=0,\\
 \EQ0:&\quad \frac{(X_{14}-X_{12})(X_{24}-X_{23})(X_{34}-X_{13})}
     {(X_{12}-X_{24})(X_{23}-X_{34})(X_{13}-X_{14})}=-1.
\end{aligned}
\end{equation}
Here $\EQ4$ follows from the first three equations in an obvious way, while $\EQ0$ is derived using their multiplicative representation

\[
 \frac{X_{i4}-X_{ij}}{X_{ij}-X_{j4}}=\frac{X_{04}-X_{0j}}{X_{0i}-X_{04}}.
\]
In order to determine the functions $X$, we compare legs in the shifted tripodal forms. The first three equations are re-written as
\begin{align*}
 T_1\EQ1: &\quad T_1\bigl(X_{02}(X_{24}-X_{23})+X_{03}(X_{23}-X_{34})
   +X_{04}(X_{34}-X_{24})=0\bigr),\\
 T_2\EQ2: &\quad T_2\bigl(X_{01}(X_{13}-X_{14})+X_{03}(X_{34}-X_{13})
   +X_{04}(X_{14}-X_{34})=0\bigr),\\
 T_3\EQ3: &\quad T_3\bigl(X_{01}(X_{14}-X_{12})+X_{02}(X_{12}-X_{24})
   +X_{04}(X_{24}-X_{14})=0\bigr),
\end{align*}
and the legs comparison leads to equations like
\[
 \a X_{03}(x_0)(X_{23}(x_2)-X_{34}(x_4))
  =X_{01}(x_0)(X_{12}(x_2)-X_{14}(x_4))+\mu(x_2)+\nu(x_4).
\]
Because of the symmetry of the indices $1,2,3,4$, we find that the functions $X_{12}$, $X_{13}$, $X_{23}$, $X_{14}$, $X_{24}$, $X_{34}$ are affinely related, and the same is true for the functions $X_{01}$, $X_{02}$, $X_{03}$, $X_{04}$. To establish a relation between the two sets of functions, it is enough to compare the legs $T_1(a^{03,04,02})$ of equation $\EQ1$ and $T_0(a^{13,14,12})$ of equation $\EQ5$:
\begin{equation}\label{logfra}
 \log\frac{X_{02}(x_2)-X_{04}(x_4)}{X_{04}(x_4)-X_{03}(x_3)}=
 \la\log\frac{X_{12}(x_2)-X_{14}(x_4)}{X_{14}(x_4)-X_{13}(x_3)}+\mu(x_2)+\nu(x_3).
\end{equation}
Differentiating with respect to $x_2$, we find:
\[
 \frac{X'_{25}(x_2)}{X_{25}(x_2)-X_{45}(x_4)}=
 \frac{\la X'_{12}(x_2)}{X_{12}(x_2)-X_{14}(x_4)}+\mu'(x_2),
\]
whence $X_{04}$ and $X_{14}$ are related by a linear-fractional transformation. Without loss of generality, we can restrict ourselves to two cases:
\begin{align*}
1)&\quad X_{ij}(x)=a_{ij}x+b_{ij},\quad i,j=0,1,2,3,4;\\
2)&\quad X_{ij}(x)=a_{ij}x+b_{ij},\quad i,j=1,2,3,4,\qquad
    X_{0i}(x)=c_i/x+d_i,\quad i=1,2,3,4.
\end{align*}
A rather tiresome analysis of the shifted tripodal forms shows that in the case 1) there are two possibilities: $a_{ij}=\a_i\a_j$, $b_{ij}=b_{12}$ or $a_{ij}=1$, $b_{ij}=\b_i+\b_j$, and in the case 2) $a_{ij}=c_ic_j$, $b_{ij}=b_{12}$, $d_i=d_1$. Non-autonomous dilations and translations allow us to reduce these cases to the quintuples (\ref{4T2-1T3}) and (\ref{4T2-1T3'}).\smallskip

{\em Equation} (\ref{Y3}). Now we can exclude from consideration equations of the type (\ref{Y2}). Then equation of the type (\ref{Y3}) can be a member of a consistent quintuple with equations of the same type only. As above, the tripodal forms (\ref{ijk}) allow us to show, starting with one equation, that a consistent quintuple must be of the form
\[
 \frac{(X_{ij}-X_{ik})(X_{kj}-X_{km})(X_{jm}-X_{im})}
      {(X_{ik}-X_{kj})(X_{km}-X_{jm})(X_{im}-X_{ij})}=-1.
\]
Comparison of the shifted tripodal forms leads to relations like (\ref{logfra}) with all possible permutations of indices, whence all $X_{ij}$ are related by linear-fractional transformations. When we encountered such a situation above, we used the method of undetermined coefficients, but, fortunately, now one can dispense with it. Indeed, the same arguments as in Lemma \ref{l:frac} allow us to deduce from (\ref{logfra}) that the linear-fractional relation between $X_{02}$ and $X_{12}$ is the same as between $X_{04}$ and $X_{14}$ or between $X_{03}$ and $X_{13}$ (with $\la=1$ and $\mu+\nu=0$). This is true for other index sets, as well, so that
\[
  X_{ik}=M_{ij}X_{jk},\quad i\ne j\ne k\ne i,
\]
where $M_{ij}$ are some linear-fractional function. There follows $M_{ij}M_{jk}=M_{ik}$, and, setting $L_i=M_{0i}$, $L_0=\id$, we obtain $M_{ij}=L_iL^{-1}_j$. But then $L^{-1}_iX_{ik}=L^{-1}_jX_{jk}$, in particular, $L^{-1}_iX_{0i}=L^{-1}_jX_{0j}$, and on the other hand, $X_{ij}=L_iX_{0j}=L_jX_{0i}$, whence $L_iL_j=L_jL_i$. Thus, $X_{ij}(x)=L_iL_jX(x)$, where all $L_i$ pairwise commute, and $X(x)$ is an arbitrary function. Setting without loss of generality $X(x)=x$, and performing a non-autonomous transformation $\tilde x(n_0,n_1,n_2,n_3,n_4)=L^{n_1}_1L^{n_2}_2L^{n_3}_3L^{n_4}_4(x(n_0,n_1,n_2,n_3,n_4))$, the situation is reduced to the case when all $L_i=\id$.\smallskip

{\em Equation} (\ref{Y7}). A simple analysis of the tripodal forms
(\ref{ijk}) and their shifted versions shows that consistent
systems of this type are reduced to linear ones with constant
coefficients, which is of course not very interesting. Note
however that in this case it makes sense to relax the condition of
the shift invariance, which might lead to auxiliary linear
problems for hypothetical 4D equations. This interesting
possibility remains outside of the scope of the present paper.
\end{proof}

\section{Remarks on the Y-system}
\label{sect: Y-eq}

We are aware of one integrable equation on $\mathbb Z^3$ which
formally belongs to the octahedron type equations and is not in
our list. The so called Y-system (of type $A_N$)
\begin{equation}\label{h6}\tag{\mbox{$Y$}}
h(m+e_1)h(m-e_1)=\frac{(1-h(m+e_2))(1-h(m-e_2))}{(1-h^{-1}(m+e_3))(1-h^{-1}(m-e_3))}
\end{equation}
plays a prominent role in the theory of integrable systems
\cite{Zamolodchikov, Adler_Startsev_1999, Doliwa_1997,
Dynnikov_Novikov_1997, FZ_2003, Kashaev_Reshetikhin_1997,
Kuniba_Nakanishi_Suzuki_2010, Szenes_2009, Volkov_2007}. The
Y-system is closely related to equations (\ref{h1})--(\ref{h5}).
However, its properties are rather different from those of
(\ref{h1})--(\ref{h5}). It has been already shown in Section
\ref{s:Hirota} that on $\mathbb Z^3$ this equation can be derived,
through suitable changes of variables, from both (\ref{h1}) and
(\ref{h2}). We discuss here in more detail its relation to
equation (\ref{h2}).

Start with a 3D Desargues (or Laplace-Darboux) map as defined in
Section \ref{s: cells}, i.e., with a map $x:Q(A_3)\to{\mathbb
R}P^n$ such that the image of any white triangle is a collinear
triple of points. In the representation of $Q(A_3)$ as $\mathbb
Z^3_{\rm even}$, this can be seen as a 2D Q-net (indexed ny
$m_1,m_2$) and its iterated Laplace transforms (indexed by $m_3$), 
see \cite{Schief_talk, Doliwa_2010}; one can find definitions
of a Q-net and of its Laplace transforms in \cite{DDG_book}. Thus, the vertices of any
white tetrahedron are collinear, and therefore one can introduce
their (real) cross-ratio, called $h$ (or $1-h$, or else $1-h^{-1}$, depending on the ordering of the arguments of the cross-ratio). Any black tetrahedron has six
white neighbors, adjacent to it along its six edges. A straightforward
computation shows that the variables $h$ for these six white tetrahedra
are related by equation (\ref{h6}) which basically says that the product
of six cross-ratios is equal to 1. To obtain it, one has to consider
four octahedra adjacent to the black tetrahedron along its black
faces. They carry four multi-ratio equations (\ref{h2}) coming from
the corresponding four complete quadrilaterals. These equations,
being multiplied, give the Y-system (the product of four multi-ratios 
can be re-arranged as a product of six cross-ratios). To sum up: variables $h$ 
of equation (\ref{h6}) live on white tetrahedra, while the equation itself 
is assigned to black tetrahedra of $Q(A_3)$.

One can try to make this picture multidimensional, by starting
with a multidimensional Desargues map $x:Q(A_N)\to{\mathbb R}P^n$.
Then any 3D sublattice $x:Q(A_3)\to{\mathbb R}P^n$ will support
its own copy of equation (\ref{h6}), with the variables $h$ which
show up in this copy and in no other: they are defined on 3-cells,
and each 3-cell belongs to only one 3D sublattice. The only
relation between these co-existing copies of equation (\ref{h6})
will be the condition that the product of $h$'s for all
tetrahedral faces of a white simplex is equal to 1.

Thus, although equation (\ref{h6}) is closely related to multidimensionally consistent
octahedron type equations, there is no reason for itself to be multidimensionally
consistent.

\paragraph{Acknowledgements.}
The research of V.A. was supported by the DFG Research Unit
``Polyhedral Surfaces'' and by RFBR grants 09-01-92431-KE and
NSh-6501.2010.2. The research of A.B. was partially supported by
the DFG Research Unit ``Polyhedral Surfaces''.



\begin{thebibliography}{99}

\bibitem{Adler_tangent}
V.E. Adler. The tangential map and associated integrable
equations. {\em J. Phys. A} {\bf 42} (2009) 332004.

\bibitem{Adler_talk}  V.E. Adler.
  Classification of discrete integrable equations of Hirota type.
  Talk at the Workshop ``Geometric Aspects of Discrete and Ultra-discrete
  Integrable Systems'', 30.03.–-03.04.2009, Glasgow, UK,
  {\tt http://www.newton.ac.uk/programmes/DIS/seminars/033014159.html}.

\bibitem{Adler_Bobenko_Suris_2003} V.E. Adler, A.I. Bobenko, Yu.B. Suris.
 Classification of integrable equations on quad-graphs. The consistency approach.
 {\em Comm. Math. Phys. \bfseries 233} (2003) {513--543}.

\bibitem{ABS_2008} V.E.~Adler, A.I.~Bobenko, Yu.B.~Suris.
 Discrete nonlinear hyperbolic equations. Classification of integrable cases.
 {\em Funct. Anal. and Appl. \bfseries 43} (2009) 3--17.

\bibitem{Adler_Startsev_1999} V.E. Adler, S.Ya. Startsev.
  Discrete analogues of the Liouville equation.
  {\em Theor. Math. Phys., \bfseries 121} (1999) 1484--1495.

\bibitem{Bobenko_talk}  A.I. Bobenko.
  From discrete differential geometry to the classification of discrete integrable
  systems.
  Talk at the Workshop ``Quantum Discrete Integrable Systems'',
  23--27.03.2009, Cambridge, UK,
  {\tt http://www.newton.ac.uk/programmes/DIS/seminars/032610006.html}.

\bibitem{DDG_book} A.I. Bobenko, Yu.B. Suris.
  {\em Discrete Differential Geometry: Integrable Structure},
  Graduate Studies in Mathematics, Vol. 98,  AMS, Providence, 2008, xxiv+404 pp.

\bibitem{Bogdanov_Konopelchenko_1998} L.V. Bogdanov, B.G. Konopelchenko.
  Analytic-bilinear approach to integrable hierarchies. II. Multicomponent KP
  and 2D Toda lattice hierarchies.
  {\em J. Math. Phys. \bfseries 39} (1998) 4701--4728.

\bibitem{CWN_1986} H.W. Capel, G.L. Wiersma, and F.W. Nijhoff.
  Linearizing integral transform for the multicomponent lattice KP.
  {\em Physica A} {\bf 138} (1986) 76--99.

\bibitem{Conway_Sloane_1991} J.H. Conway, N.J.A. Sloane.
  The cell structures of certain lattices. In: Miscellanea mathematica, P. Hilton, F.
  Hirzebruch, and R. Remmert (eds.), pp. 71--107, Springer, 1991.

\bibitem{Doliwa_1997} A. Doliwa.
  Geometric discretisation of the Toda system.
  {\em Phys. Lett. A} {\bf 234}  (1997) 187--192.

\bibitem{Doliwa_2010} A. Doliwa.
  Desargues maps and the Hirota-Miwa equation.
  {\em Proc. Royal Soc. A} {\bf 466}  (2010) 1177--1200.

\bibitem{Doliwa_prepr} A. Doliwa.
  The affine Weyl group symmetry of Desargues maps and of the non-commutative
  Hirota-Miwa system.
  {\tt arXiv:1006.3380 [nlin.SI]}.


\bibitem{Dorfman_Nijhoff_1991} I.Ya. Dorfman, F.W. Nijhoff.
  On a (2+1)-dimensional version of the Krichever-Novikov equation.
  {\em Phys. Lett. A\bfseries 157} (1991) 107--112.

\bibitem{Dynnikov_Novikov_1997} I.~A. Dynnikov and S.~P. Novikov.
  Discrete spectral symmetries of small-dimensional differential operators and difference operators on regular lattices and two-dimensional manifolds.
  {\em Russian Math. Surveys} {\bf 52} (1997) 1057--1116.

\bibitem{FG_2006}
V. Fock, A. Goncharov. Moduli space of local systems and higher
Teichm\"uller theory, {\em Publ. Math. Inst. Hautes \'Etudes Sci.}
{\bf 103 (1)} (2006).

\bibitem{FZ_2003} S. Fomin, A. Zelevinsky.
  Y-systems and generalized associahedra.
  {\em Ann. of Math.} {\bf  158} (2003) 977--1018.

\bibitem{Henriques_2007} A. Henriques.
  A periodicity theorem for the octahedron recurrence.
  {\em J. Algebr. Comb.} {\bf 26} (2007) 1--26.

\bibitem{Hilbert-Cohn-Vossen} D. Hilbert, S. Cohn-Vossen.
  {\em Geometry and the imagination}. Chelsea Publishing Company,
  New York, 1952, ix+357 pp.

\bibitem{Hirota_1981} R.~Hirota.
  Discrete analogue of a generalized Toda equation.
 {\em J. Phys. Soc. Japan} {\bf 50} (1981), 3785--3791.

 \bibitem{Kashaev_Reshetikhin_1997} R. Kashaev, N. Reshetikhin.
   Affine Toda field theory as a 3-dimensional integrable system.
  {\em Comm. Math. Phys.} {\bf 188} (1997) 251--266.

\bibitem{Kazakov_2008} V. Kazakov, A. Sorin, A. Zabrodin.
   Supersymmetric Bethe Ansatz and Baxter Equations from Discrete Hirota Dynamics.
  {\em Nucl.Phys.B} {\bf 790} (2008) 345--413.

\bibitem{King_Schief_2003} A.D. King, W.K. Schief.
  Tetrahedra, octahedra and cubo-octahedra: integrable geometry of multi-ratios.
  {\em J. Phys. A \bfseries 36:3} (2003) {785--802}.

\bibitem{Konopelchenko_Schief_2002} B.G. Konopelchenko, W.K. Schief.
  Menelaus' theorem, Clifford configurations and inversive geometry of the
  Schwarzian KP hierarchy.
  {\em J. Phys. A \bfseries 35:29} (2002) {6125--6144}.


\bibitem{KLWZ_1997} I. Krichever, O. Lipan, P. Wiegmann, and A. Zabrodin.
  Quantum integrable models and discrete classical Hirota equations.
  {\em Comm. Math. Phys.} {\bf 188} (1997) 267--304.

\bibitem{Kuniba_Nakanishi_Suzuki_2010} A. Kuniba, T. Nakanishi, and J. Suzuki.
  T-systems and Y-systems in integrable systems.
  {\tt arXiv: 1010.1344 [hep-th]}, 155 pp.

\bibitem{LNQ_2009} S.B. Lobb, F.W. Nijhoff, and G.R.W. Quispel.
  Lagrangian multiform structure for the lattice KP system.
  {\em J. Phys. A} {\bf 42} (2009) 472002 (11 pp).

\bibitem{Miwa_1982} T.~Miwa.
  On Hirota's difference equations.
  {\em Proc. Japan Acad.}, Ser. A (Math. Sci.) \textbf{58} (1982), 9--12.

\bibitem{Moody_Patera_1992} R.V. Moody, J. Patera.
  Voronoi and Delaunay cells of root lattices: classification of their facets by
  Coxeter–Dynkin diagrams.
  {\em J. Phys. A} {\bf 25} (1992) 5089--5134.

\bibitem{Nijhoff_Capel_1990} F.W. Nijhoff, H.W. Capel.
  The direct linearisation approach to hierarchies of integrable PDEs in $2 + 1$
  dimensions: I. Lattice equations and the differential-difference hierarchies.
  {\em Inv. Problems \bfseries 6} (1990) 567--590.

\bibitem{NCW_1985} F.W. Nijhoff, H.W. Capel, and G.L. Wiersma.
  Integrable lattice systems in two and three dimensions.
  In: Geometric Aspects of the Einstein Equations and Integrable Systems,
  Ed. R. Martini, {\em Lecture Notes in Physics}, Berlin/New York, Springer Verlag,
  1985, pp. 263--302.

\bibitem{NCWQ_1984} F.W. Nijhoff, H.W. Capel, G.L. Wiersma, and G.R.W. Quispel.
   B\"acklund transformations and three-dimensional lattice equations.
   {\em Phys. Lett. A} {\bf 105} (1984) 267--272.

\bibitem{Schief_talk}  W.K. Schief.
  Discrete Laplace–Darboux sequences, Menelaus’ theorem and the pentagram map.
  Talk at the Workshop ``Geometric Aspects of Discrete and Ultra-discrete Integrable
  Systems'',
  30 March–3 April 2009, Glasgow, UK,
  {\tt http://www.newton.ac.uk/programmes/DIS/seminars/040309309.html}.

\bibitem{Shinzawa_2000} N. Shinzawa.
  Symmetric linear B\"acklund transformation for discrete BKP and DKP equations.
  {\em J. Phys. A} {\bf 33} (2000) 3957--3970.

\bibitem{Shinzawa_Saito_1998} N. Shinzawa, S. Saito.
  A symmetric generalization of linear B\"acklund transformation associated with
  the Hirota bilinear difference equation.
  {\em J. Phys. A: Math. Gen.} {\bf  31} (1998) 4533--4540.

\bibitem{Speyer_2007} D.E. Speyer.
  Perfect matchings and the octahedron recurrence.
  {\em J. Algebr. Comb.} {\bf 25} (2007) 309--348.

\bibitem{Szenes_2009} A. Szenes.
  Periodicity of Y-systems and flat connections.
  {\em Lett. Math. Phys.} {\bf 89} (2009) 217--230.

\bibitem{Volkov_2007} A.Yu. Volkov.
  On the periodicity conjecture for Y-systems.
  {\em Commun. Math. Phys.} {\bf 276} (2007) 509--517.

\bibitem{Zabrodin_1997} A.V. Zabrodin.
  Hirota’s difference equations.
  {\em Theor. Math. Phys.} {\bf 113} (1997) 1347--1392.

\bibitem{Zamolodchikov}
A.B.~ Zamolodchikov. On the thermodynamic Bethe ansatz equations
for reflectionless ADE scattering theories. {\em Phys. Lett. B}
{\bf 253} (1991) 391--394.

\end{thebibliography}
\end{document}